\documentclass[aps, prx, reprint, amsmath,amssymb,showpacs,floatfix,longbibliography, twocolumn, superscriptaddress,nofootinbib]{revtex4-2}
\usepackage{times}
\usepackage{graphicx}
\usepackage{dcolumn}
\usepackage{bm}
\usepackage{color}
\usepackage{xcolor}
\usepackage{hyperref}
\usepackage{enumitem}
\usepackage{cancel}
\usepackage{amsthm}
\usepackage{CJKutf8}
\usepackage{stmaryrd}
\usepackage{suffix}
\usepackage{tablefootnote}
\usepackage{eczoo}
\usepackage{bbm} 
\usepackage{booktabs}

\usepackage{comment}
\hypersetup{colorlinks=true,citecolor=blue,linkcolor=blue, urlcolor=blue}
\hypersetup{linktocpage}
\newcolumntype{M}[1]{>{\centering\arraybackslash}m{#1}}
\newcolumntype{N}{@{}m{0pt}@{}}
\usepackage{environ}
\usepackage{MnSymbol}

\bibpunct{[}{]}{,}{n}{}{}

\usepackage{amsfonts,amssymb,amsmath}
\usepackage[T1]{fontenc}
\usepackage{tikz}
\newcommand{\bra}[1]{\langle {#1} |}
\newcommand{\ket}[1]{ | {#1} \rangle}
\let\originalleft\left
\let\originalright\right
\renewcommand{\left}{\mathopen{}\mathclose\bgroup\originalleft}
\renewcommand{\right}{\aftergroup\egroup\originalright}
\newcommand{\Mod}[1]{\ \mathrm{mod}\ #1}

\NewEnviron{eqs}{%
\begin{equation}\begin{split}
    \BODY
\end{split}\end{equation}
}

\usepackage{blindtext}
\newtheorem{theorem}{Theorem}[section]

\newtheorem{example}{Example}
\newtheorem{proposition}{Proposition}

\newcommand{\Yijia}[1]{ { \color{blue} (Yijia: {#1}) }}

\newcommand{\vva}[1]{ { \color{red} (VVA: {#1}) }}
\WithSuffix\newcommand\vva*[1]{{\color{red} #1}}
\WithSuffix\newcommand\Yijia*[1]{{\color{blue} #1}}

\def\prg#1{\paragraph*{{\bf #1}}}

\begin{document}
\begin{CJK*}{UTF8}{gbsn}

\title{Clifford operations and homological codes for rotors and oscillators}

\author{Yijia Xu (许逸葭)}
\thanks{Equal contribution.}

\affiliation{Joint Center for Quantum Information and Computer Science, NIST and University of Maryland, College Park, Maryland 20742, USA}
\affiliation{Joint Quantum Institute, National Institute of Standards and Technology
NIST and University of Maryland, College Park, Maryland 20742, USA}
\affiliation{Institute for Physical Science and Technology, University of Maryland, College Park, Maryland 20742, USA}

\author{Yixu Wang (王亦许)}
\thanks{Equal contribution.}
\affiliation{Institute for Advanced Study, Tsinghua University, Haidian District, Beijing, China, 100084}
\affiliation{Joint Center for Quantum Information and Computer Science, NIST and University of Maryland, College Park, Maryland 20742, USA}

\author{Victor V. Albert}
\affiliation{Joint Center for Quantum Information and Computer Science, NIST and University of Maryland, College Park, Maryland 20742, USA}

\begin{abstract}

We develop quantum information processing primitives for the planar rotor, the state space of a particle on a circle.
By interpreting rotor wavefunctions as periodically identified wavefunctions of a harmonic oscillator, we determine the group of bosonic Gaussian operations inherited by the rotor.
This \(n\)-rotor Clifford group, \(\text{U}(1)^{n(n+1)/2} \rtimes \text{GL}_n(\mathbb{Z})\), is represented by continuous \(\text{U}(1)\) gates generated by polynomials quadratic in angular momenta, as well as discrete \(\text{GL}_n(\mathbb Z)\) momentum sign-flip and sum gates.
We classify homological rotor error-correcting codes [\href{https://doi.org/10.48550/arXiv.2303.13723}{arXiv:2303.13723}] and various rotor states based on equivalence under Clifford operations.

Reversing direction, we map homological rotor codes and rotor Clifford operations back into oscillators by interpreting occupation-number states as rotor states of non-negative angular momentum.
This yields new multimode homological bosonic codes protecting against dephasing and changes in occupation number, along with their corresponding encoding and decoding circuits.
In particular, we show how to non-destructively measure the oscillator phase using conditional occupation-number addition and post selection.
We also outline several rotor and oscillator varieties of the GKP-stabilizer codes [\href{https://doi.org/10.48550/arXiv.1903.12615}{arXiv:1903.12615}].

\end{abstract}
\date{\today}
\maketitle

\end{CJK*}

\section{Introduction}

Quantum information processing can be done in various quantum platforms, which can be described in the abstract by one of only a few state spaces.
For example, the qudit state space describes any few-level physical system, irrespective of the physical nature of the levels.
Similarly, a continuous-variable harmonic oscillator state space models vibrations in ions and materials, as well as photons confined to cavities or traveling through various media.

Conventional quantum state spaces come with a set of primitives --- canonical states and operations --- that are both physically motivated and essential for information processing schemes.
For the case of a qudit, canonical operations come from the Pauli group as well as its normalizer, the Clifford group; canonical states are eigenstates of the Pauli operators \cite{gottesman1998heisenberg,tolar2018clifford}.
 For the case of the oscillator, the operations are the oscillator displacements and more general quadratic Gaussian (a.k.a.\ Bogoliubov) operations, while the canonical states are states of fixed position or momentum \cite{lloyd2012gaussian,serafini2017quantum}.

\begin{figure}[t]
    \centering
    \includegraphics[width=0.45\textwidth]{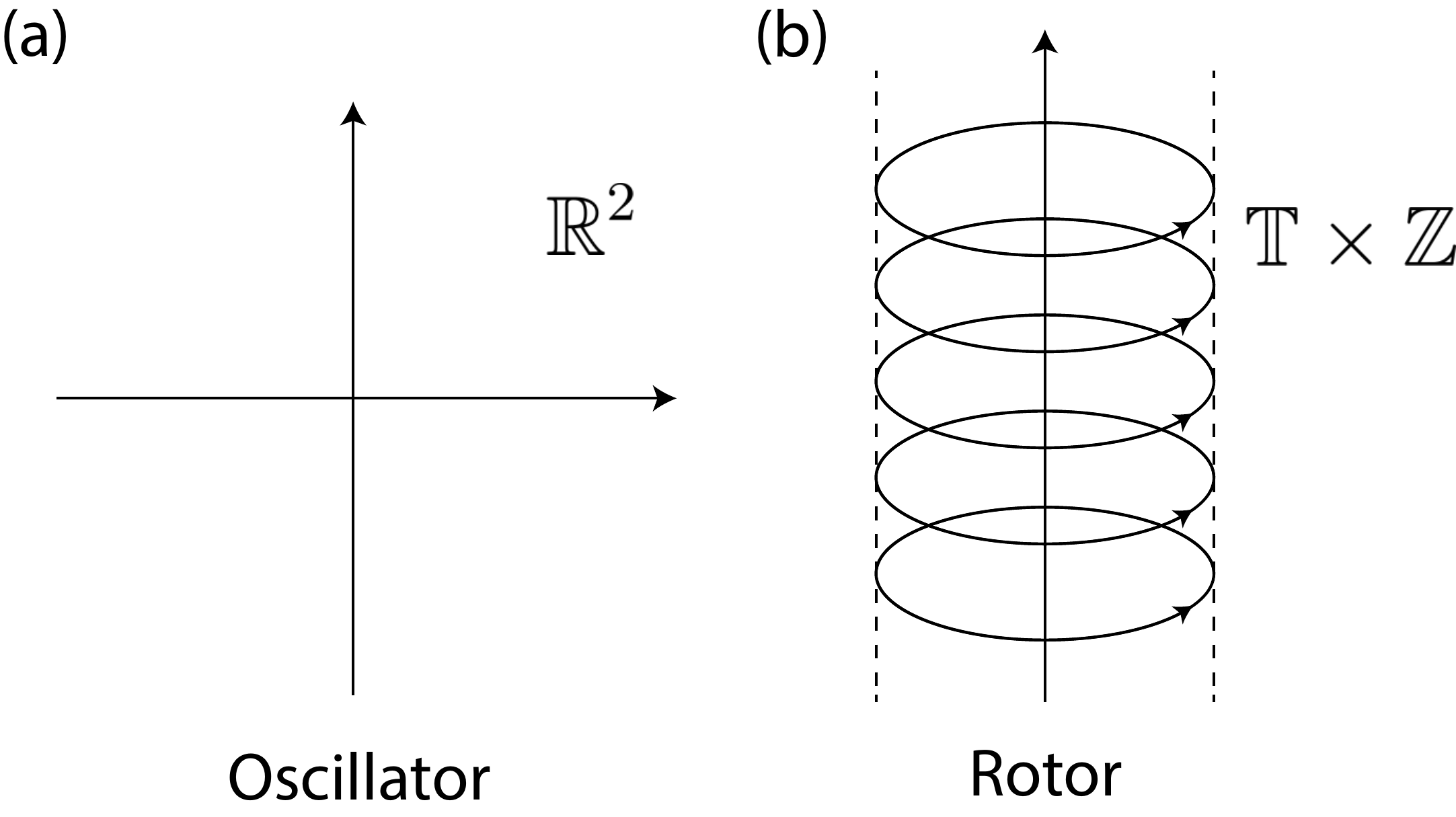}
    \caption{Comparison between the phase space of (a) oscillator and (b) rotor systems. The dashed line denotes that the rotor angular momenta are confined to the integers.}
    \label{fig:phase_space}
\end{figure}

\begin{figure*}[t]
    \centering
    \includegraphics[width=0.9\textwidth]{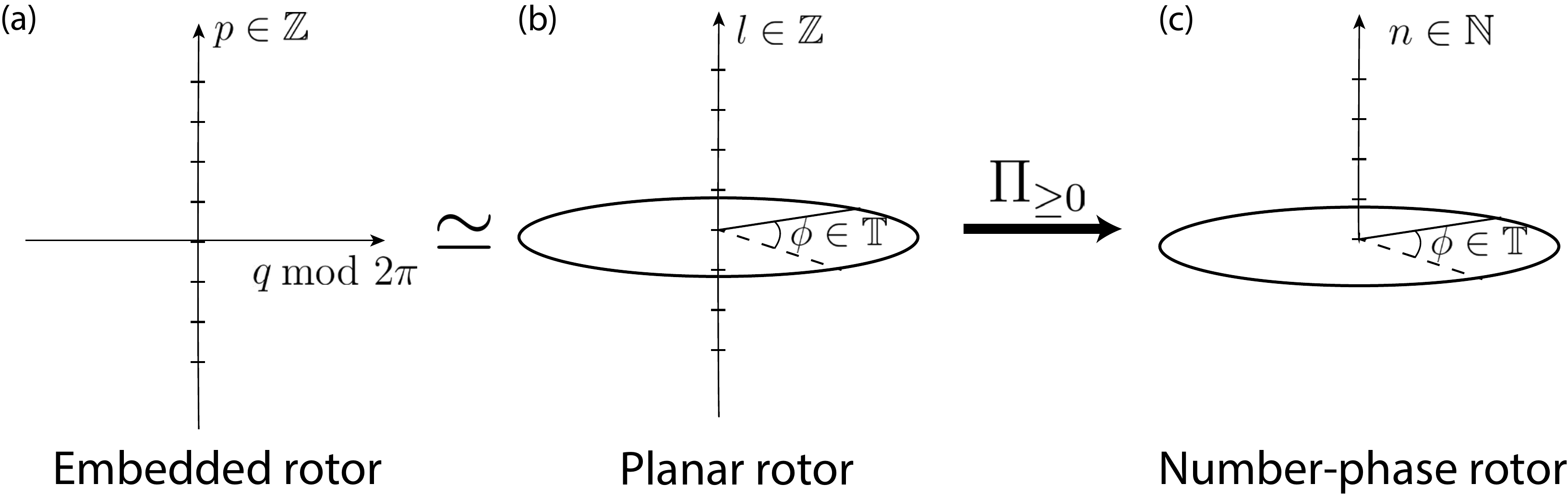}
    \caption{
    The planar rotor (b) describes the motion of a particle on a circle \(\mathbb T\), with said particle admitting only integer-valued angular momenta.
    This phase space can be embedded into that of the harmonic oscillator by periodically identifying oscillator positions \(q\), yielding the embedded rotor (a).
    Alternatively, interpreting oscillator Fock states as the positive-momentum states of a rotor yields the number-phase interpretation of the oscillator (c).
    Relations between these constructions allow us to identify primitives necessary for information processing as well as develop various error-correcting codes for all three state spaces.
    }
    \label{fig:rotors}
\end{figure*}

While few-level systems and harmonic oscillators have been well studied for over 100 years, a third ``angular'' state space --- the planar or \(\text{U}(1)\) rotor \cite{carruthers1968phase,levy1976afraid,mukunda1979wigner,floreanini1991quantum,kowalski1996coherent,zhang2003phase,chadzitaskos2012quantizations,przanowski2014weyl,kastrup2016wigner,kowalski2021wigner,mivsta2022angle,gazeau2022integral, Mista2022angle}--- has lagged behind in its development.
For example, while the types of primitive rotor operations have long been known \cite{juan2014normalizer,bermejo2016normalizer,albert2017general,albert2020robust}, the specific group formed by them is, to our knowledge, not yet established.

A planar rotor describes the state space of a quantum system confined to a circle.

Such systems have been overlooked in the past due to a lack of controllable quantum platforms amenable to a rotor description.
However, rotor systems are gaining traction due to recent improvements in the control of
superconducting circuits \cite{koch2007charge,schreier2008suppressing},  ultracold molecules \cite{anderegg2019optical,liu2019molecular}, ion traps \cite{urban2019coherent,glikin2023systematic}, free electrons \cite{reinhardt2021free,dahan2023creation}, and orbital angular momentum systems~\cite{rigas2008full,rigas2010non,rigas2011orbital,kastrup2016wigner,kastrup2017wigner_operator,kastrup2017wigner_qi}.

The recent exciting experimental progress warrants a deeper quantitative investigation into the rotor's information processing primitives.
We perform this investigation in this work, obtaining a firmer understanding of rotor states and operations.
We also obtain several immediate applications to quantum error correction of both oscillator and rotor platforms.

We enumerate canonical primitives of the planar rotor and, in particular, determine the group formed by its canonical unitary operations.
We do so by treating the rotor --- whose configuration space is described by an angle --- as a subspace of the harmonic oscillator that is defined by periodically identifying oscillator positions.
This embedding of the rotor into the oscillator allows us to view rotor primitives as oscillator primitives that preserve the embedded rotor subspace.
Our embedded rotor treatment also provides insight into quantum error correction, yielding a rotor version of a recent class of oscillator error-correcting codes \cite{noh2020encoding} and revealing several connections between oscillator and rotor codes.

In parallel, we focus on the oscillator itself, developing its rotor-like ``polar'' interpretation in terms of occupation number and phase degrees of freedom \cite{susskind1964quantum,carruthers1968phase,pegg1989phase}.
This interpretation does not quite yield a bona-fide rotor, which has led to much pontificating in the literature \cite{lynch1995quantum,gour2002quantum}
Nevertheless, we show that many of the primitives of the planar rotor can be transferred into this number-phase rotor interpretation of the oscillator in a lossless fashion.
Other primitives are somewhat distorted, but are still present in this number-phase ``rotor''. For example, certain unitary operations become non-unitary when mapped into the number-phase rotor, but nevertheless remain valid and useful quantum channels.

Our development of number-phase primitives allows us to map several classes of error-correcting codes, including the recently developed homological rotor codes \cite{vuillot2023homological}, into the number-phase ``rotor''.
These new oscillator codes are polar analogues of lattice codes \cite{gottesman2001encoding, noh2020encoding} \eczoo{quantum_lattice} and are compatible with oscillator noise channels where photon loss is present but random-rotation (i.e., dephasing) noise is dominant.
As such, we anticipate that these codes will be relevant to trapped-ion systems \cite{glikin2023systematic, fluhmann2019encoding}.

\subsection{Summary of results}

Positions of a planar rotor are labeled by an angle, which makes for a compact configuration space.
The dual basis of momentum states is labeled by the integers \(\mathbb{Z}\), making for a discrete and infinite set of labels called the circle group \(\mathbb{T}\cong \text{U}(1)\).
A planar rotor can thus be thought of as being ``in between'' the qudit and oscillator, encapsulating both the compactness of the former and the infinite-dimensional nature of the latter.

\prg{Rotor Clifford group}

The \(n\)-qubit Clifford group consists of all unitary operations that preserve Pauli-matrix commutation relations, and we exclude the Pauli group from this definition for simplicity \cite{gottesman1998heisenberg}.
Similarly, the analogous \(n\)-oscillator group of Gaussian transformations consists of all operations that preserve the commutation relations between position and momentum \cite{bartlett2002efficient,ferraro2005Gaussian,serafini2017quantum}.
In both cases, the commutation-preservation can be tied to preservation of a particular symplectic form, and the two groups correspond to the symplectic groups \(\text{Sp}_{2n}(\mathbb{Z}_2)\) and \(\text{Sp}_{2n}(\mathbb{R})\), respectively.

The rotor Clifford group is the group of unitary operations that preserve commutation relations between rotor position and momentum shift Pauli-type operators.
It has been studied before in the context of efficient simulation \cite{bermejo2016normalizer}, and its generators have been detailed before \cite{bermejo2012classical,juan2014normalizer,bermejo2016normalizer,albert2020robust}.
However, the \textit{structure} of this group has, to our knowledge, not yet been identified.
This is, in part, due to the fact that a rotor's angular positions and momenta are labeled by \textit{different} types of numbers --- angles vs.\ integers --- which complicates analogous symplectic formulations.

By periodically extending planar rotor wavefunctions such that they form a subset of oscillator wavefunctions, we observe that the rotor Clifford group can be thought of as a subgroup of the oscillator symplectic group, $\text{Sp}_{2n}(\mathbb{R})$.
Projecting this group into the state subspace of the rotor, we find the \(n\)-rotor Clifford group to be a semi-direct product of two groups --- the Lie group $\text{U}(1)^{n(n+1)/2} $ and the discrete group $\text{GL}_n(\mathbb{Z})$ of unimodular integer-valued invertible matrices,
\begin{eqs}
    \text{Rotor Clifford group}=\text{U}(1)^{n(n+1)/2} \rtimes \text{GL}_n(\mathbb{Z})~.
\end{eqs}
The Lie group corresponds to gates generated by products of two rotor momenta, while the discrete group is generated by conditional momentum shifts and momentum sign flips.
We summarize our results in Table~\ref{table:group_for_oscillator_and_rotor}.

\prg{Rotor ``Gaussian'' states}
A key primitive of the harmonic oscillator is the set of Gaussian states \cite{lloyd2012gaussian} --- states which can be characterized exclusively by expectation values of linear and quadratic polynomials in position and momentum.
At the edge of this set are position and momentum states, which admit infinitely precise expectation values of the two quadratures.

Projecting oscillator position and momentum states into the rotor subspace and removing any states whose projected versions are no longer Gaussian shows that only the momentum states survive.
Similarly, projecting a class of Gaussian states called coherent states yields a class of previously developed rotor coherent states \cite{kowalski1996coherent,hall2002coherent, gonzalez1998coherent,rigas2011orbital, gazeau2022integral} (cf. \cite{chadzitaskos2012quantizations, kowalski2021wigner, gulshani1979generalized})

We study how coherent states transform under Clifford operations, revealing that the Clifford subgroup that preserves their structure is the \textcolor{black}{signed} permutation group (in contrast to the much larger group of passive linear-optical transformations in the case of the oscillator \cite{lloyd2012gaussian}).
While squeezed oscillator states do not make the cut when projected into the embedded rotor, they can nevertheless be simulated by adding an extra ``regularization'' parameter \cite{menicucci2014fault} to the rotor coherent states, yielding states equivalent to those in Ref.~\cite{kowalski2002uncertainty}.

We also take a look at how the canonical rotor Hamiltonian describing a Josephson junction \cite{girvin2014circuit} transforms under Clifford operations, revealing an interplay of Gaussian and stabilizer transformations.

\prg{Classifying homological rotor codes}

Equipped with better understanding of the rotor Clifford group and the embedded rotor construction, we investigate the structure of the homological rotor codes \cite{vuillot2023homological} \eczoo{homological_rotor} --- a recent extension of stabilizer codes \cite{gottesman1997stabilizer} \eczoo{qubit_stabilizer}
to rotors.

We classify homological rotor codes using the Smith normal form --- the standard tool for homology-group calculation.
We show that the Smith normal form of a code cannot be changed to that of another code by any rotor Clifford operation that preserves the code's CSS structure \eczoo{css}.
The different Smith normal forms thus label different equivalence classes of homological rotor codes under such operations.
If Clifford operations are the only available operations for an encoding, this implies that such operations have to act on a resource state within the same Smith class. These resource states are tightly related to oscillator Gottesman-Kitaev-Preskill (GKP) states \cite{gottesman2001encoding} \eczoo{gkp}.

\prg{Homological number-phase codes}

Returning to the harmonic oscillator, we study multi-mode extensions of the number-phase codes \cite{arne2020rotation,hillmann2022performance,endo2022quantum} \eczoo{number_phase} --- polar analogues of bosonic lattice codes \eczoo{quantum_lattice} that protect against occupation-number loss/gain and dephasing, which correspond to distortions in the oscillator's number and phase degrees of freedom, respectively. The error correction scheme and performance of random rotation-symmetric codes are recently studied in Ref.~\cite{marinoff2023explicit, totey2023performance}.

Bosonic rotation codes are defined for a single oscillator, and multimode extensions have not been substantially studied \cite{jain2023quantum}.
We show how to map the entire class of homological rotor codes into the oscillator, yielding a new class of polar-like codes protecting against loss and dephasing noise.

We also show that the rotor Clifford-group encodings of homological rotor codes can be performed by analogous operations in the number-phase picture, but with some caveats.
Mapping rotor Clifford group transformation into the number-phase interpretation of the oscillator yields a Clifford \textit{semigroup} consisting of some non-unitary transformations.
For example, while quadratic momentum gates are mapped to a Kerr interaction, conditional momentum shifts are mapped to conditional photon injection --- a non-unitary operation.

Nevertheless, all Clifford semigroup operations are valid quantum channels, and the original rotor Clifford algebra is mostly left intact after the number-phase mapping.
As a result, a Clifford-based encoding of homological codes is still possible.

On the other hand, extraction of error syndromes for number-phase codes becomes more difficult.
Nevertheless, we show it is possible to measure the effect of an oscillator rotation \textit{non-destructively} with the help of post-selection.
This provides a probabilistic error recovery alternative to Knill tele-correction for bosonic rotation codes \cite{grimsmo2021quantum} and may be relevant to metrological protocols for determining the oscillator phase in a non-desctructive fashion.

\prg{Other new codes and relations}

The relation between planar rotors, their embeddings into the oscillator, and the number-phase interpretation of the oscillator (see Fig.~\ref{fig:rotors}) allows us to treat various seemingly unrelated error-correcting codes in the same fashion.

We outline how GKP-stabilizer codes \cite{noh2020encoding} \eczoo{gkp-stabilizer} can be mapped into the number-phase degrees of freedom of the oscillator, yielding another class of codes protecting a (possibly infinite) logical state space against loss and dephasing noise.

The embedded rotor construction enables us to investigate the underlying connections between rotor and oscillator codes. In particular, for an embedded rotor, a single homological rotor code with torsion is equivalent to an oscillator GKP code encoding a qudit. Moreover, rotor GKP codes \cite{gottesman2001encoding,raynal2010encoding,albert2020robust} \eczoo{rotor_gkp}, can be included in the same framework, as a concatenation of homological rotor codes and modular-qudit GKP codes \cite[Sec.~II]{gottesman2001encoding} \eczoo{qudit_gkp} (see Example \ref{example:torison_GKP}). While homological number-phase codes are multi-mode generalizations of the original number-phase codes \cite{arne2020rotation}, they can also be viewed as a rotation-symmetric generalization of multi-mode GKP codes. These relations are illustrated in Fig.~\ref{fig:code_relation}.

\subsection{Outline of the manuscript}

In Sec.~\ref{sec:preliminary}, we introduce the generalized Pauli operators for planar rotors. We describe a method to embed a logical planar rotor into a single-mode harmonic oscillator, as well as the number-phase interpretation of the oscillator.

Next, we investigate the generators of rotor Clifford group and their symplectic representation in Sec.~\ref{sec:Clifford_generator} and identify its group structure in Sec.~\ref{sec:Clifford_group_structure}. In Sec.~\ref{sec:non_exist_squeezing}, we show that a unitary squeezing operation cannot exit for rotors, in contrast to oscillator systems.
In Sec.~\ref{sec:passive_transformation_group}, we give the subgroup of rotor Clifford group that is also a subgroup of the group of Gaussian transformations that preserve the total occupation number.

In Sec.~\ref{sec:rotor_Gaussian_state}, we study ``Gaussian'' states of rotors and their behavior under Clifford group and passive subgroup, including nullifier states (Sec.~\ref{sec:rotor_nullifier_states}) and coherent states (Sec.~\ref{sec:rotor_coherent_state}). We also discuss how the Clifford group transforms the Josephson-junction Hamiltonian in Sec.~\ref{sec:josephson_junction}.

We then revisit the formalism of homological rotor codes and investigate the physical implications of torsion and Smith normal form by calculating the codes' homology group in Sec.~\ref{sec:torsion}.
In Sec.~\ref{sec:equivalence_class}, we show that the codes with different torsion parts from different equivalent classes which cannot be related by CSS Clifford transformations. In Sec.~\ref{sec:rotor_gkp}, we show that rotor GKP codes are concatenations between homological rotor codes and modular-qudit codes.

In Sec.~\ref{sec:homo_np_rotor_code}, starting from an Example.~\ref{example:np_as_projected_rotor} that number-phase codes are rotor GKP codes after projecting on the non-negative angular momentum subspace, we propose the homological number-phase codes --- a multi-mode generalization of number-phase codes. In Prop.~\ref{pros:existance_homo_np}, we provide a procedure of mapping homological rotor codes to homological number-phase codes.
 
 In Sec.~\ref{sec:clifford_semigroup_encoding}, we demonstrate how to use the Clifford semigroup of number-phase operations to encode in homological number-phase codes.

 In Sec.~\ref{sec:gkp_repetition_rotor}, we show that GKP-stabilizer codes can be generalized to $\text{U}(1)$ rotor systems as well as number-phase rotor systems and compare their differences.

In the appendices, we collect several miscellaneous results related to quantum applications of rotors.
In Appendix~\ref{app:coherent_Wigner}, we calculate the Wigner function for the rotor GKP codewords and show they indeed have negativity. In Appendix~\ref{app:ec_condition_rotor_GKP}, we calculate the error correction conditions for the normalized rotor GKP codes and discuss their relations with Jacobi $\vartheta$ functions. In Appendix~\ref{app:nogo_rotor} we discuss an analogy of Gaussian encoding no-go theorem \cite{vuillot2019toric} for rotors.

\section{The planar rotor and friends}\label{sec:preliminary}

In this section, we review the setup of the planar, or $\text{U}(1)$, quantum rotor and its various connections to the harmonic oscillator (see Fig.~\ref{fig:phase_space})~\cite{bermejo2016normalizer,albert2017general,albert2020robust,albert2022bosonic}.

\subsection{Planar rotor}

The state space of a rotor is the same as that of a particle on the group $\text{U}(1)$, arising naturally from a quantum body rotating in two dimensions. The state space admits bases of fixed particle position and fixed particle momentum.
We associate the former with the ``Pauli $X$-basis'' of the rotor, denoting basis elements by a phase $\theta \in [0,2\pi)$, which also label elements of the circle group,
\begin{eqs}
    \text{U}(1)\cong\mathbb{T}\cong\mathbb{R}/2\pi~.
\end{eqs}
Conversely, the dual ``Pauli $Z$-basis'' is characterized by irreducible representations of $\text{U}(1)$, which are labeled by the integers $\mathbb{Z}$. The two bases are
\begin{eqs}
    X\text{-basis}&:~~ \ket{\theta}, \quad \theta\in \mathbb{T},\\
    Z\text{-basis}&:~~ \ket{l}, \quad l \in \mathbb{Z},
\end{eqs}
where we call $\ket{\theta}$ as a vector in phase state, and $\ket{l}$ as an angular momentum state.
The former can be expressed in terms of the latter via the Fourier series,
\begin{eqs}\label{eq:rotor-position}
    \ket{\theta}=\frac{1}{\sqrt{2\pi}}\sum_{\ell\in\mathbb{Z}}e^{i\theta\ell}\ket{\ell}~,
\end{eqs}
and visa versa.

The fundamental operators on a single rotor are the generalized Pauli operators. The Pauli $X$ operator is parameterized by an integer $m \in \mathbb{Z}$, and the Pauli $Z$ operator is parameterized by a phase factor $\phi \in \mathbb{T}$. Their actions on the angular position and momentum states are as follows:
\begin{eqs}\label{eq:rotor_pauli_on_state}
    & X(m) \ket{\theta}=e^{im\theta} \ket{\theta}, X(m)\ket{l}=\ket{l+m}, \quad m\in \mathbb{Z},\\
    & Z(\phi) \ket{\theta} = \ket{\theta-\phi}, Z(\phi) \ket{l}= e^{i\phi l } \ket{l}, \quad \phi\in \text{U}(1).
\end{eqs}
These are natural generalizations of qudit Pauli operators, defined on the space of a particle on the group $\mathbb{Z}_q$, or harmonic oscillator displacement operators defined for $\mathbb{R}$.

We would like to present the Pauli operators in terms of the fundamental degrees of freedom phases and angular momentum.
Though the angular momentum operator $\hat{l}$ is well-defined, we cannot define phase operator $\hat{\theta}$ individually because of the ambiguity introduced by the $2\pi$ periodicity. Nevertheless, we can still express Pauli operators as exponentials of either operator,
\begin{eqs}\label{eq:rotor_pauli_operator}
    X(m)= e^{im\hat{\theta}}, \quad Z(\phi)= e^{i \phi \hat{l}}.
\end{eqs}
The commutation relation between Pauli $X$ and $Z$ for rotor is
\begin{eqs}\label{eq:rotor_pauli_algebra}
    X(m)Z(\phi)=e^{-im\phi} Z(\phi) X(m),
\end{eqs}
where the phase factor arises from the rotor commutation relation $[\hat{l}, e^{i \hat{\theta}}]= e^{i \hat{\theta}}$.
Tensor products of these Pauli operators generate the \(n\)-rotor Pauli group, \(\mathcal{P}^{\text{rot}}_n\).

\subsection{Embedded rotor}

A planar rotor can be embedded into a harmonic oscillator by periodically identifying the oscillator's positions [see Fig.~\ref{fig:rotors}(a-b)].
It was pointed out in Ref.~\cite{vuillot2023homological} that this embedding can be done by restricting to the \(+1\) eigenspace of one of the stabilizers of the GKP code. This stabilizer, $\hat{S}_q=\exp(i 2\pi \hat{p})$, written in terms of the oscillator momentum \(\hat p\), shifts the oscillator position \(q\) by \(2\pi\). The stabilizer constraint,
\begin{eqs}
    \hat{S}_q\ket{\psi}&=\int dq \psi(q) \ket{q-2\pi} \overset{!}{=} \int dq \psi(q) \ket{q}=\ket{\psi}~,
\end{eqs}
restricts the possible oscillator wavefunctions to those satisfying $\psi(q+2\pi)=\psi(q)$.
In other words, imposing the above stabilizer constraint is equivalent to imposing \(2\pi\)-periodicity in the $q$ representation.\footnote{In terms of the oscillator subsystem decomposition from Ref.~\cite{pantaleoni2020modular}, the above performs the decomposition \(\mathbb{T}\otimes 2\pi\mathbb{Z}=\mathbb{R}\), where $\mathbb{T}$ is the embedded rotor space, and the factor $2\pi \mathbb{Z}$ is called the gauge space.
}

The stabilizer \(\hat{S}_q\) shifts each oscillator position state by \(2\pi\), meaning that its eigenstates are superpositions of a position state with all states related to it by a shift of a multiple of \(2\pi\).
Such states are labeled by an angle and correspond to the position states of the embedded rotor,
\begin{eqs}
  |\theta_{\text{emb}}\rangle=\sum_{m\in\mathbb{Z}}\left|q=2\pi m+\theta\right\rangle~.
\end{eqs}
Since they are eigenstates of a GKP stabilizer, these are of the same comb-like form as the logical codewords of the GKP code.

An oscillator momentum state is an eigenstate of \(\hat{S}_q\) only when the momentum \(p\) is an integer.
Such states can be expressed in terms of embedded-rotor position states via a Fourier transform,
\begin{eqs}
  |p\rangle&\propto\int dqe^{ipq}\left|q\right\rangle \\&=\sum_{\ell\in\mathbb{Z}}\int d\theta e^{ip\left(2\pi\ell+\theta\right)}\left|2\pi\ell+\theta\right\rangle \\&=\int d\theta e^{ip\theta}|\theta_{\text{emb}}\rangle~.
\end{eqs}
Note that the integer momentum assumption is necessary for the last equality to hold.

Pauli operators for the embedded rotor,
\begin{eqs}\label{eq:embedded_rotor_pauli}
    X(m)&=e^{im\hat{q}}, ~m\in \mathbb{Z},\\
    Z(\phi)&=e^{i \phi \hat{p}}, ~\phi \in \mathbb{R}~,
\end{eqs}
where \(\hat q\) is the oscillator position operator,
form the subset of oscillator displacements that preserve the embedded rotor space.
Because of the periodicity in $q$, the momentum $p$ is quantized and can only take integer values. As a logical operator, $X(m)$ must commute with stabilizer $\hat{S}_q$, so $m$ can only take integer values.
The oscillator commutation relation $[\hat{q},\hat{p}]=i$ yields that of the embedded rotor,
\begin{eqs}
    e^{im \hat{q}} e^{i \phi \hat{p}}= e^{-im \phi} e^{i \phi \hat{p}} e^{im \hat{q}}~.
\end{eqs}

The embedded rotor helps us understand various features of the planar rotor from the harmonic oscillator perspective.
Since this space inherits oscillator quadrature operators, both its position and momentum operators are well-defined,
in contrast to a real quantum rotor in which a standalone $\hat{\theta}$ is not well-defined.
This makes the embedded rotor ``stretchable'', meaning that any periodicity can be imposed in lieu of \(2\pi\). In other words, if we squeeze the oscillator quadrature which embeds a rotor, we are essentially redefining the embedded rotor space without disturbing any logical information within the rotor (see Sec.~\ref{sec:rotor_Clifford} and Sec.~\ref{sec:rotor_Gaussian_state}).

\subsection{Number-phase rotor}

A different and, in this case, lossy way to embed the \(\text{U}(1)\) rotor into an oscillator is to associate rotor momentum states with oscillator Fock bases, $\ket{n},~n \in \mathbb{N}$.
Since $\mathbb{N} \subset \mathbb{Z}$, we can view an oscillator as a rotor that is \textit{projected} onto the non-negative angular momentum subspace. We can then associate rotor position states from Eq.~\eqref{eq:rotor-position} with the Pegg-Barnett oscillator phase states \cite{susskind1964quantum,carruthers1968phase,pegg1989phase},
\begin{eqs}\label{eq:oscillator_phase_state}
    \ket{\theta}_{\text{np}}&=\sum_{n \in \mathbb{N}} e^{-in\theta} \ket{n},~\theta \in \mathbb{T}~.\\
\end{eqs}
These states play important roles in quantum error correction \cite{grimsmo2021quantum,hillmann2022performance} and quantum metrology \cite{holevo2011probabilistic,iosue2022continuous}, and are relevant in the construction of continuous-variable designs \cite{iosue2022continuous}.
Indeed, both $\ket{\theta}_{\text{np}},\forall \theta \in \mathbb{T}$ and $\ket{n}, \forall n \in \mathbb{N}$ form a  complete (but, in the former case, non-orthonormal) basis.
We denote this re-interpretation of the oscillator as a \textit{number-phase rotor} [see Fig.~\ref{fig:rotors}(b-c)].

We define quadrature operators of the number-phase rotor using the projection
\begin{eqs}
  \Pi_{\geq m}=\sum_{l \geq m} \ket{l} \bra{l}
\end{eqs}
for \(m=0\); this projection removes all negative momentum states.
We denote all projected operators by the subscript ``np'', i.e., \(O_{\text{np}} = \Pi_{\geq 0} O \Pi_{\geq 0}\) for any operator \(O\).

The number-phase rotor Pauli operators are
 \begin{eqs}\label{eq:fake_rotor_quadrature}
   X(m)_{\text{np}}&=\begin{cases}
   \sum_{n=0}^{\infty}\ket{n+m}\bra{n} & m\geq0\\
   \sum_{n=0}^{\infty}\ket{n}\bra{n+|m|} & m<0
   \end{cases}\\Z(\phi)_{\text{np}}&=e^{i\phi\hat{n}}\,,
  \end{eqs}
where $\hat{n}=\sum_{l=0}^\infty l\ket{l} \bra{l} $ should be interpreted as the photon number operator.

In the number-phase interpretation, $X(m)_{\text{np}}$ performs \(m\)-photon injection (subtraction) for positive (negative) \(m\), and both correspond to powers of the Kogut-Susskind phase operator and its adjoint \cite{susskind1964quantum,bergou1991operators,bartlett2002quantum}.
The $Z(\phi)_{\text{np}}$ operator is a single-mode rotation by $\phi$, generated by the harmonic oscillator Hamiltonian.

Because for $m\geq 0$, we have
\begin{eqs}
    X(m)_{\text{np}}^\dagger X(m)_{\text{np}}&=\mathbb{I},\\
    \quad X(m)_{\text{np}} X(m)_{\text{np}}^\dagger &=\Pi_{\geq m} \neq \mathbb{I},
\end{eqs}
hence $X(m)_{\text{np}}$ is not a unitary. The intuition behind this is that one can inject and then subtract the same number of photons on arbitrary states. But subtracting and then injecting $m$ photons can only be applied to states that have at least $m$ photons to begin with. Overall, $X(m)_{\text{np}}$ and $Z(\phi)_{\text{np}}$ generate a semigroup with an identity element, i.e., a monoid.

The number-phase treatment of an oscillator as a rotor without negative angular momentum eigenstates turns out to be quite fruitful, as much of the algebra structure of rotors survives under this projection \cite{albert2022bosonic}, e.g., the Pauli commutation relation,
\begin{eqs}
  X(m)_{\text{np}}Z(\phi)_{\text{np}}=e^{-im\phi}Z(\phi)_{\text{np}}X(m)_{\text{np}}\,.
\end{eqs}
The similarities are sufficient to enable us to map homological rotor codes into sensible codes for number-phase rotors.
We return to this ``rotor'' in Sec.~\ref{sec:homo_np_rotor_code}.

\section{Rotor Clifford group}\label{sec:rotor_Clifford}

The qubit Clifford group is a collection of operations which map a tensor product of Pauli operators to another tensor product of Pauli operators while also preserving the commutation relations among them.
To be specific, given the \(n\)-qubit Pauli group $\mathcal{P}_n$, the corresponding \(n\)-qubit Clifford group is
\begin{eqs}\label{eq:clifford_quotient}
    C(\mathcal{P}_n)=N_{\text{U}(2^n)}(\mathcal{P}_n)/\mathcal{P}_n~,
\end{eqs}
where $N_{\text{U}(2^n)}$ is the group of all elements of $\text{U}(2^n)$ that preserve the Pauli group \cite{tolar2018clifford}.
The denominator denotes that we study the group elements up to (right) multiplication by an element of the Pauli group,
and also ignore any phases.
Clifford groups for qudits, oscillators, and rotors are defined similarly.
These form a subgroup of ``easy'' gates for quantum computation, but cannot be used to obtain an arbitrary logical gate.

From a quantum error correction point of view, the Clifford group plays a fundamental role in designing and characterizing quantum error-correcting codes \cite{haah2013commuting,haah2016algebraic,haah2021clifford,haah2022topological,haah2023nontrivial}, as well as the equivalence and deformation of codes \cite{chen2023equivalence,bonilla2021xzzx,dua2022clifford,huang2023tailoring}. Investigations of the Clifford group also yield efficient classical algorithms for simulating Clifford circuits, guaranteed by the Gottesman-Knill theorem \cite{gottesman1998heisenberg} and its generalization to arbitrary abelian groups \cite{juan2014normalizer,bermejo2016normalizer}.
On the practical side, the underlying mathematical structure of circuit quantization and its connections to symplectic transformations have also recently been studied \cite{egusquiza2022algebraic,parra2022quantum,parra2023geometrical,rymarz2023consistent,osborne2023symplectic}.

A nice property of Clifford groups is that their elements can be expressed as symplectic transformations acting on particular vectors --- \(2n\)-dimensional vectors $(X|Z) \in \mathbb{Z}_d^{2n}$ for an \(n\)-qudit system, and \(2n\)-dimensional real vectors $(\Vec{q}|\Vec{p}) \in \mathbb{R}^{2n}$ for an \(n\)-mode system.
Since the domain of both quadrature pairs are the same in both cases, it is easy to show that the qudit and oscillator Clifford groups are $\text{Sp}_{2n}(\mathbb{Z}_d)$ \cite{hosten2005stabilizer} and $\text{Sp}_{2n}(\mathbb{R})$ \cite{lloyd2012gaussian, serafini2017quantum}, respectively.
Since rotor systems are \textit{hybrid} --- behaving like continuous-variable systems in the phase basis and discrete-variable systems in the momentum basis --- the rotor Clifford group is not as easy to read off.

In this section, we first present the generators of the Clifford group for $n$ rotors, highlighting the reason why the squeezing operator does not exist in the rotor system.
We then identify the \(n\)-rotor Clifford group to be the semi-direct product group $\text{U}(1)^{n(n+1)/2} \rtimes \text{GL}_n(\mathbb{Z})$.
We also delineate the subgroup of the Clifford group that preserves the structure of rotor coherent states, facilitating the discussion in Section \ref{sec:rotor_coherent_state}.

\begin{table*}[t]
\centering

\begin{tabular}{cccc}
\toprule
 & Oscillators \cite{ferraro2005Gaussian,serafini2017quantum} & Planar rotors (this work) & Qubits \cite{gottesman1998heisenberg}\tabularnewline
\midrule
Displacement group & ~~~$\text{H}_{3}(\mathbb{R})^{n}$ (Heisenberg-Weyl)~~~ & ${\cal P}_{n}^{\text{rot}}$ (rotor Pauli \cite{aldaya1996algebraic}) & ${\cal P}_{n}$ (Pauli)\tabularnewline
\midrule
Symplectic group & $\text{Sp}_{2n}(\mathbb{R})$ & ~~~~~$\text{U}(1)^{n(n+1)/2} \rtimes\text{GL}_{n}(\mathbb{Z})$ (rotor Clifford)~~~~~ & $\text{Sp}_{2n}(\mathbb{Z}_{2})$ (Clifford)\tabularnewline
\midrule
Passive symplectic group & $\text{U}(n)=\text{Sp}_{2n}(\mathbb{R})\cap\text{SO}(2n)$ & $\mathbb{Z}_2 \wr S_n $ & $S_n$ \tabularnewline
\bottomrule
\end{tabular}

\caption{Canonical groups for oscillator, rotor, and qubit systems.
The displacement group shifts the canonical variables of each system, while the symplectic group preserves the displacement group.
The passive symplectic group of the oscillator is the symmetry group of the complex sphere formed by \(n\)-mode coherent states of the same energy, while its rotor and qubit counterparts are intersections of this group with the respective rotor and qubit symplectic groups.
}
\label{table:group_for_oscillator_and_rotor}
\end{table*}

\subsection{Clifford group generators}
\label{sec:Clifford_generator}

All planar rotor Clifford gates can be expressed as exponentials of quadratic combinations of the phase and angular momentum operators, except for the parity-flip operation. In the following, we define the generators of the rotor Clifford group \cite{oppenheim1997signals} and show how they transform rotor Pauli operators.

\begin{subequations}
\label{eq:rotor-generators}
\begin{enumerate}
    \item The \textsc{cnot} gate is defined as
    \begin{eqs}\textsc{cnot}_{1\rightarrow 2}= \int_{\text{U}(1)} Z(\phi) \otimes \ket{\phi} \bra{\phi} d\phi=\sum_{l\in\mathbb{Z}} \ket{l} \bra{l} \otimes X(l),
    \end{eqs}
    and it acts on Pauli operators as
        \begin{eqs}\label{eq:rotor_cnot}
            &\textsc{cnot}_{1\rightarrow 2} (X(1) \otimes \mathbb{I}) \textsc{cnot}_{1\rightarrow 2}^\dagger =X(1) \otimes X(1)\\
        &\textsc{cnot}_{1\rightarrow 2} (\mathbb{I} \otimes X(1)) \textsc{cnot}_{1\rightarrow 2}^\dagger=\mathbb{I}\otimes X(1)\\
        &\textsc{cnot}_{1\rightarrow 2} (\mathbb{I} \otimes Z(\phi)) \textsc{cnot}_{1\rightarrow 2}^\dagger =Z(-\phi)\otimes Z(\phi)\\
        &\textsc{cnot}_{1\rightarrow 2} (Z(\phi) \otimes \mathbb{I}) \textsc{cnot}_{1\rightarrow 2}^\dagger=Z(\phi) \otimes \mathbb{I}~.
        \end{eqs}

    \item The \textsc{quad} gate, $\textsc{quad}_\varphi= e^{i\varphi\hat{l}(\hat{l}+1)/2}$, acts on Pauli \(X\) operators as
    \begin{eqs}\label{eq:rotor_quad}
        \textsc{quad}_\varphi X(1)\textsc{quad}_\varphi ^\dagger= Z(\varphi) X(1)~,
    \end{eqs}
    commuting with all Pauli $Z$ operators.

    \item The \textsc{cphs} gate, $\textsc{cphs}_\varphi=e^{i \varphi \hat{l} \otimes \hat{l}}$, commutes with \(Z\) operators and acts on \(X\) operators as
    \begin{eqs}\label{eq:rotor_cphs}
        \textsc{cphs}_\varphi (X(1)\otimes \mathbb{I}) \textsc{cphs}_\varphi^\dagger= X(1) \otimes Z(\varphi),\\
        \textsc{cphs}_\varphi (\mathbb{I}\otimes X(1)) \textsc{cphs}_\varphi^\dagger= Z(\varphi) \otimes X(1).
    \end{eqs}

    \item The parity flip,  $\textsc{p} = \sum_\ell \left|-\ell\right\rangle\langle\ell| = \textsc{p}^\dagger$, acts as
    \begin{eqs}
        &\textsc{p} X(m) \textsc{p}= X(-m),\\
        &\textsc{p} Z(\phi) \textsc{p}= Z(-\phi)~,
    \end{eqs}
    flipping the sign of both position and momentum.
\end{enumerate}
\end{subequations}
All of the above operators can be generated by evolving under quadratic interactions, with the notable exception of the parity flip.
Nevertheless, the flip is a well-defined Clifford operation that can be realized in concrete systems (see Sec.~\ref{sec:josephson_junction}) and that plays a critical role in the study of the relation between rotor and number-phase rotor codes (see Sec. \ref{sec:homo_np_rotor_code}).

Despite the hybrid nature of rotor systems, a Pauli $Z(\Vec{\boldsymbol{\phi}}) X(\Vec{\boldsymbol{m}})$ can still be represented by a vector $\Vec{\boldsymbol{v}}$ \cite[Lemma 2.8]{prasad2008decomposition,bermejo2016normalizer},
\begin{eqs}\label{eq:rotor_pauli_vector}
    \Vec{\boldsymbol{v}}=\begin{pmatrix}
        \Vec{\boldsymbol{m}}_{\boldsymbol{v}}^T
|\Vec{\boldsymbol{\phi}}_{\boldsymbol{v}}^T
    \end{pmatrix}^T,
\end{eqs}
where $\Vec{\boldsymbol{m}}_{\boldsymbol{v}}$ is a $n$-dimensional integer-valued column vector, and $\Vec{\boldsymbol{\phi}}_{\boldsymbol{v}}$ is a $n$-dimensional $\mathbb{T}$-valued column vector.
The commutation relation between two Pauli strings $\Vec{\boldsymbol{u}}$ and $\Vec{\boldsymbol{v}}$ can then be represented as
\begin{eqs}
    &Z(\Vec{\boldsymbol{\phi}}_{\boldsymbol{u}}) X(\Vec{\boldsymbol{m}}_{\boldsymbol{u}}) Z(\Vec{\boldsymbol{\phi}}_{\boldsymbol{v}} ) X(\Vec{\boldsymbol{m}}_{\boldsymbol{v}})\\
    &= e^{-i \Vec{\boldsymbol{u}}^T \Lambda \Vec{\boldsymbol{v}}} Z(\Vec{\boldsymbol{\phi}}_{\boldsymbol{v}}) X(\Vec{\boldsymbol{m}}_{\boldsymbol{v}}) Z(\Vec{\boldsymbol{\phi}}_{\boldsymbol{u}}) X(\Vec{\boldsymbol{m}}_{\boldsymbol{u}}),
\end{eqs}
where $e^{-i \Vec{\boldsymbol{u}}^T \Lambda \Vec{\boldsymbol{v}}}$ is the phase factor captured by the symplectic inner product between $\Vec{\boldsymbol{u}}$ and $\Vec{\boldsymbol{v}}$,
\begin{eqs}
    &\Vec{\boldsymbol{u}}^T \Lambda \Vec{\boldsymbol{v}}=\Vec{\boldsymbol{m}}^T_{\boldsymbol{u}} \Vec{\boldsymbol{\phi}}_{\boldsymbol{v}} -\Vec{\boldsymbol{\phi}}^T_{\boldsymbol{u}} \Vec{\boldsymbol{m}}_{\boldsymbol{v}},\quad \Lambda= \begin{pmatrix}
        0 & \mathbb{I}_{n\times n} \\
        -\mathbb{I}_{n\times n} &0 \end{pmatrix}~.
\end{eqs}

Each Clifford circuit $U$ can be represented by a symplectic matrix $Q_U$ that transforms a Pauli string $\Vec{v}$ as
\begin{eqs}\label{eq:symplectic_coefficient}
    U \Vec{\boldsymbol{v}} U^\dagger =Q_U\Vec{\boldsymbol{v}}=\begin{pmatrix}
        \Vec{\boldsymbol{m}}_{Q_U \boldsymbol{v}}^T|\Vec{\boldsymbol{\phi}}_{Q_U \boldsymbol{v}}^T
    \end{pmatrix}^T~,
\end{eqs}
for some quadrature transformation \(Q_U\).
Since \(U\) is a Clifford circuit, $\Vec{\boldsymbol{m}}_{Q_U \boldsymbol{v}}$ should be a $n$-dimensional integer-valued column vector, and $\Vec{\boldsymbol{\phi}}_{Q_U \boldsymbol{v}}$ should be a $n$-dimensional $\mathbb{T}$-valued column vector. Any quadrature transformation $Q$ also has to satisfy
\begin{eqs}\label{eq:symplecitc}
    Q^T \Lambda Q=\Lambda\quad\quad\text{(symplectic condition)}
\end{eqs}
because it preserves the rotor commutation relations.
This implies that the rotor Clifford group is a particular subgroup of the oscillator Clifford group $\text{Sp}_{2n}(\mathbb{R})$ that preserves the angle-integer form of the Pauli vectors $\vec{\boldsymbol{v}}$.

\subsection{Clifford group structure}\label{sec:Clifford_group_structure}

Most generally, a rotor quadrature transformation can be written as
\begin{eqs}\label{eq:rotor_symplectic}
    Q=\begin{pmatrix}
        Q_{XX} & 0\\
        Q_{XZ} & Q_{ZZ}
    \end{pmatrix}~.
\end{eqs}
The upper-left block is a general linear transformation over the integers, \(Q_{XX} \in \text{GL}_n(\mathbb{Z})\); all such transformations have determinant \(\pm1\) (i.e., are \textit{unimodular}).
The upper-right block must be zero because a sum of a continuous variable and a discrete variable is not discrete, thereby violating the discreteness of \(\Vec{\boldsymbol{m}}\).
The lower-left block does not have to be zero since there are no discreteness constraints on the angles \(\Vec{\boldsymbol{\phi}}\).

Imposing the symplectic condition \eqref{eq:symplecitc} yields the following constraints on the block matrices:
\begin{eqs}\label{eq:symplectic_condition}
Q_{ZZ}^{T}Q_{XX}=Q_{XX}^{T}Q_{ZZ}&=\mathbb{I}_{n\times n}\\
Q_{XX}^{T}Q_{XZ}-(Q_{XX}^{T}Q_{XZ})^{T}&\equiv0\Mod2\pi
\end{eqs}
The first equality of Eq.~\eqref{eq:symplectic_condition} requires that $Q_{ZZ}=(Q_{XX}^{T})^{-1}$, indicating that the bottom-right block $Q_{ZZ}$ also has to be unimodular. The second equality requires $Q_{XX}^T Q_{XZ}$ to be symmetric. \textcolor{black}{These requirements are derived directly from the definition of the rotor Clifford group. They constrain the most general form of the rotor Clifford group element to be
\begin{eqs}\label{eq:rotor_symplectic_rep}
    Q=\begin{pmatrix}
        Q_{XX} & 0\\
        (Q_{XX}^T)^{-1}C & (Q_{XX}^T)^{-1}
    \end{pmatrix}~,
\end{eqs}
in which $Q_{XX} \in \text{GL}_n(\mathbb{Z})$ and $C$ is a \ensuremath{\mathbb{T}}-valued symmetric matrix whose addition generates $\text{U}(1)^{n(n+1)/2}$.}

\textcolor{black}{The generators of the symplectic group from Eqs.~\eqref{eq:rotor-generators} are sufficient to generate any rotor Clifford group element in the form of Eq.~\eqref{eq:rotor_symplectic_rep}.}
They fall into two classes: \textit{block diagonal} generators that do not mix phase and angular-momentum degrees of freedom, and \textit{block off-diagonal} generators that do mix them.
These two respective sets naturally generate two subgroups of symplectic transformations,\footnote{Although it is known that $\text{GL}_n(\mathbb{Z})$ can be generated by as few as $2$ generators for any order of $n\geq 4$, we use \textsc{cnot} and $\textsc{p}$ because they are local gates that are physically preferred.}
\begin{eqs}\label{eq:clifford_generator}
H&=\Bigg\{\begin{pmatrix}A & 0\\
0 & (A^{T})^{-1}
\end{pmatrix} \Bigg|,~A\in\text{GL}_{n}(\mathbb{Z})\Bigg\}=\left\langle \textsc{cnot},\textsc{p}\right\rangle ,\\N&=\Bigg\{\begin{pmatrix}\mathbb{I}_{n\times n} & 0\\
C & \mathbb{I}_{n\times n}
\end{pmatrix} \Bigg|,~C~\text{ is \ensuremath{n\times n} \ensuremath{\mathbb{T}}-valued symmetric matrix}\Bigg\},\\&=\left\langle \textsc{quad},\textsc{cphs}\right\rangle .
\end{eqs}
The matrices in the \textit{diagonal} or \textit{CSS} subgroup $H$ form a (reducible) representation of $\text{GL}_n (\mathbb{Z}) $; this group does not mix positions with momenta.
The subgroup $N$ is the addition group of $n\times n$ symmetric matrices over $\mathbb{T}$, representing the Lie group $\text{U}(1)^{n(n+1)/2} $.

It is straightforward to identify how these generators conjugate a Clifford operator defined by the symplectic transformation \(Q\); we tabulate their actions in Table.~\ref{table:clifford_Gate}.

\begin{table*}[t]
\centering
\begin{tabular}{ |p{2cm}|p{6cm}|p{6cm}|  }

\hline
\multicolumn{3}{|c|}{Matrix representation of rotor Clifford gates} \\
\hline
Gate& Left multiplication & Right multiplication \\
\hline
$\textsc{swap}_{ij}$ & Swap the $i$th row and the $j$th row of $Q_{XX}$, and swap the $i$th row and the $j$th row of $(Q_{XZ}|Q_{ZZ})$. & Swap the $i$th column and the $j$th column of $Q_{XX}$, and swap the $i$th column and the $j$th column of $Q_{XZ}$ and $Q_{ZZ}$. \\
\hline
$\textsc{cnot}_{ij}$ & Add the $i$th row of $Q_{XX}$ to the $j$th row of $Q_{XX}$, and subtract the $j$th row of $Q_{ZZ}$ from the $i$th row of $Q_{ZZ}$. & Add the $j$th column of $Q_{XX}$ to the $i$th column of $Q_{XX}$, and subtract the $i$th column of $Q_{ZZ}$ from the $j$th column of $Q_{ZZ}$.  \\
\hline
$\textsc{p}_i$    &Multiply the $i$th rows of $Q_{XX}, Q_{XZ}, Q_{ZZ}$ by $-1$ & Multiply the $i$th rows of $Q_{XX}, Q_{XZ}, Q_{ZZ}$ by $-1$ \\
\hline
$\textsc{quad}_{\varphi,i}$ & Add $\varphi$ times the $i$th row of $Q_{XX}$ to the $i$th row of $Q_{XZ}$   &  Add $\varphi$ times the $i$th column of $(Q_{XX}^T)^{-1}$ to the $i$th row of $Q_{XZ}$ \\
\hline
$\textsc{cphs}_{\varphi,ij}$ &Add $\varphi$ times the $i$th row of $Q_{XX}$ to the $j$th row of $Q_{XZ}$, and add $\varphi$ times the $j$th row of $Q_{XX}$ to the $i$th row of $Q_{XZ}$  & Add $\varphi$ times the $i$th column of $(Q_{XX}^T)^{-1}$ to the $j$th column of $Q_{XZ}$, and add $\varphi$ times the $j$th column of $(Q_{XX}^T)^{-1}$ to the $i$th column of $Q_{XZ}$     \\

\hline
\end{tabular}
\caption{The matrix description of rotor Clifford gates. The swap gate can be realized by $\textsc{swap}_{ij}=\textsc{p}_j\textsc{cnot}_{ji}^\dagger\textsc{cnot}_{ij}^\dagger\textsc{cnot}_{ji}$.}
\label{table:clifford_Gate}
\end{table*}

Having mapped all generators of the Clifford group into matrix form, we prove that the combination of the above two representations forms the following group.

\begin{theorem}
    Rotor symplectic transformations form the group \(\text{U}(1)^{n(n+1)/2} \rtimes\emph{GL}_n(\mathbb{Z})\).
\end{theorem}

\begin{proof}

We first prove that the two subgroups form a group and then proceed with showing that \(N\) is normal.

\begin{enumerate}
    \item \textit{Associativity}: The associativity of group $G$ follows from the associativity of matrix multiplication.\\

\item \textit{Identity}: The identity element of group $G$ is
        \begin{eqs}
            e= \begin{pmatrix}
                \mathbb{I}_{n\times n } & 0\\
                0 & \mathbb{I}_{n\times n }
            \end{pmatrix},
        \end{eqs}
     which is the identity matrix.

\item \textit{Inverse}: We first prove that $\forall g\in G$, $\exists \mathcal{A} \in H$ and $\mathcal{C} \in N $ such that $g=\mathcal{A}\mathcal{C}$. Without loss of generality, a generic element $g$ of $G$ can be written as product of matrices in $H$ and $N$ as
     \begin{eqs}
         g=\mathcal{A}_1 \mathcal{C}_1 \cdots \mathcal{A}_n \mathcal{C}_n
     \end{eqs}
for some constant \(n\).
where $\mathcal{A}_1,...,\mathcal{A}_n \in H$ and $\mathcal{C}_1,...,\mathcal{C}_n \in G$. This can be proved by induction. The $k=1$ case is obvious. Suppose for $k=n$, $g=\mathcal{A}_1 \mathcal{C}_1 \cdots \mathcal{A}_n \mathcal{C}_n=\mathcal{A}^{(n)}\mathcal{C}^{(n)}$. Then $g^{\prime}=\mathcal{A}_1 \mathcal{C}_1 \cdots \mathcal{A}_{n+1} \mathcal{C}_{n+1}=\mathcal{A}^{(n)} \mathcal{C}^{(n)}\mathcal{A}_{n+1} \mathcal{C}_{n+1}=\mathcal{A}^{(n+1)}\mathcal{C}^{(n+1)}$, where $\mathcal{A}^{(n+1)}=\mathcal{A}^{(n)}\mathcal{A}_{n+1}$ and $\mathcal{C}^{(n+1)}=\mathcal{A}_{n+1}^{-1}\mathcal{C}^{(n)}\mathcal{A}_{n+1}\mathcal{C}_{n+1}$. To show $\mathcal{C}^{(n+1)}\in N$ we need to show $\mathcal{A}_{n+1}^{-1}\mathcal{C}^{(n)}\mathcal{A}_{n+1}\in N$, which is proved in Eq.~\eqref{eq:normality_clifford}.

Then the existence of the inverse of $g\in G$ follows from the existence of the matrix inverse of the element of $\mathcal{A}\in H$ and $\mathcal{C}\in N$, and $g^{-1}=\mathcal{C}^{-1}\mathcal{A}^{-1}$.

    \begin{widetext}
\item \textit{Normality}: $\forall Q\in G$ and $\forall h\in N$,
\begin{eqs}\label{eq:normality_clifford}
    QhQ^{-1}=
    &\begin{pmatrix}
        Q_{XX} & 0\\
        Q_{XZ} & (Q_{XX}^T)^{-1}
    \end{pmatrix}  \begin{pmatrix}
        \mathbb{I}_{n \times n} & 0\\
        C_h & \mathbb{I}_{n \times n}
    \end{pmatrix}
    \begin{pmatrix}
        Q_{XX}^{-1} & 0\\
        -Q_{XZ} & Q_{XX}^T
    \end{pmatrix}\\
    =& \begin{pmatrix}\mathbb{I}_{ n\times n} & 0\\
    [Q_{XZ} Q_{XX}^{-1} + (Q_{XX}^{-1})^T C_h Q_{XX}^{-1}- (Q_{XX}^{-1})^T Q_{XZ}] \Mod 2\pi & \mathbb{I}_{n\times n}\end{pmatrix}\\
    =& \begin{pmatrix}\mathbb{I}_{ n\times n} & 0\\
     [(Q_{XX}^{-1})^T C_h Q_{XX}^{-1} ] \Mod 2\pi & \mathbb{I}_{n\times n}\end{pmatrix} = \begin{pmatrix}\mathbb{I}_{ n\times n} & 0\\
     C' & \mathbb{I}_{n\times n}\end{pmatrix}  \in N.
\end{eqs}
\end{widetext}
In the last line of Eq.~\eqref{eq:normality_clifford}, since $(Q_{XX}^{-1})^T C_h Q_{XX}^{-1}$ is a symmetric matrix, $C^{\prime}\equiv[(Q_{XX}^{-1})^T C_h Q_{XX}^{-1}] \Mod 2\pi$ is also a $\mathbb{T}$-valued symmetric matrix.
Above derivation uses the identity
\begin{eqs} \label{eq:modular_arithmetic}
    (A (C \Mod 2\pi))\Mod 2\pi = (AC) \Mod 2\pi,
\end{eqs}

where $A $ is a unimodular matrix and $C$ is a $\mathbb{T}$-valued matrix. This is because this matrix multiplication composes of multiplications of modular variable with integers and summations of two modular variables. Due to the fundamental property of modulation, $ (c+d)\Mod 2\pi = [(c \Mod 2\pi)+(d \Mod 2\pi)] \Mod 2\pi$, $\forall c,d\in \mathbb{T}$, and $ [a (b \Mod 2\pi)]\Mod 2\pi =(ab) \Mod 2\pi$, $ \forall a\in \mathbb{Z},~b\in\mathbb{T}$, above equation is true.

\end{enumerate}

\end{proof}

\subsection{No squeezing for planar rotors}\label{sec:non_exist_squeezing}

In this subsection, we recap that (unitary) squeezing is not a Clifford operation for rotors because it is an automorphism of $\mathbb{R}^{n}$ but not of $\mathbb{T}^{n}$ or $\mathbb{Z}^{n}$.

The squeezing of a single-mode harmonic oscillator squeezes one quadrature and dilates its conjugate pair, \(\hat{q}\rightarrow e^{r}\hat{q}\) and \(\hat{p}\rightarrow e^{-r}\hat{p}\) for some real parameter \(r\).
It is a symplectic automorphism in $\mathbb{R}$ so is an element of the oscillator Clifford group.
Notably, squeezing preserves the spectrum of both the position and momentum operators, which is all of the reals.

For rotor systems, a map multiplying the angular position by a constant \(c\), \(\theta \to c\theta\), can only be an automorphism for \(c = \pm1\). Hence, squeezing a rotor is impossible.

For multiple rotors, the automorphism of $\mathbb{T}^n$ is $\text{GL}_n(\mathbb{Z})$. Because any matrix of $\text{GL}_n(\mathbb{Z})$ has determinant  $\pm 1$, there should not exist an overall squeezing in the phase basis. Because a $\text{GL}_n(\mathbb{Z})$ matrix can be diagonalized via unimodular matrices, and unimodular matrices have determinant $\pm 1$, this implies that the resulting diagonal matrix should also have determinant $\pm 1$. Since the entries of it should be integers, the eigenvalues can only be $\pm 1$. This means for multiple rotor modes, there does not exist an automorphism that squeezes one collective mode while stretches another collective mode under the phase basis.

The angular momentum basis has group structure $\mathbb{Z}^n$ and its group automorphism group is $\text{GL}_n(\mathbb{Z})$ as well. The same argument shows that squeezing is not an automorphism of this basis either.

Not having unitary squeezing is a generic feature in other systems as well, e.g., modular-qudit systems whose states are valued in $\mathbb{Z}_d$. More generally, if one of the quadratures is valued in a compact group, squeezing should not be able to be realized via a unitary operation, as it cannot be an automorphism of the group.

\subsection{Passive symplectic subgroup}\label{sec:passive_transformation_group}

The oscillator symplectic group has an important subgroup --- the group of \textit{passive} transformations that preserve the total energy (i.e., photon number) of the oscillators \cite{serafini2017quantum}.
\textcolor{black}{The passive symplectic group for rotor systems can be similarly defined as preserving the total energy of the rotors. For $n$ identical rotors, the total energy should be proportional to the sum of angular momentum squared of each rotor, $\sum_i l_i^2$.}
This yields the corresponding passive symplectic group of the rotor.

The passive symplectic group of the oscillator preserves the total photon number, $\sum_i \hat a_i^\dagger \hat{a}_i=n/2+\sum_i(\hat q_i^2+\hat p_i^2)$.
Collecting positions and momenta into a \(2n\)-dimensional vector \(\vec v\), we see that passive transformations have to preserve the inner product \(\vec v^T \vec v\).
This constraint defines a \(2n\)-dimensional real sphere in phase space, meaning that any passive transformation has to be an element of the sphere's proper-rotation symmetry group, \(\text{SO}(2n)\).
Taking the intersection of this group with the symplectic group yields $\text{Sp}_{2n}(\mathbb{R})\cap\text{SO}(2n) \cong \text{U}(n)$, the \(n\)-dimensional unitary group.

The real \(2n\)-dimensional sphere can equivalently be thought of as a complex \(n\)-dimensional sphere, whose corresponding constraint can be formulated using the vector of annihilation operators, \(\vec a = (\hat a_1,\hat a_2,...,\hat a_n)\).
Passive transformations form that sphere's symmetry group, \(\text{U}(n)\).
Since coherent states are eigenstates of annihilation operators, they can be thought of as points on said sphere.
Passive transformations rotate these points, preserving the tensor product structure of coherent states on different modes,
\begin{eqs}
    U\ket{\alpha_1}\otimes...\ket{\alpha_n}=\ket{\tilde\alpha_1}\otimes...\ket{\tilde\alpha_n}~,
\end{eqs}
where $\tilde{\alpha}_i=\sum_j U_{ij}\alpha_j$.

In the symplectic representation, a general passive symplectic element can be written as
\begin{eqs}
    S=\begin{pmatrix}
        A & B\\
        -B & A
    \end{pmatrix},~~ \text{where}~~~ \begin{aligned}
        & AB^T-BA^T=0,\\
        & AA^T+BB^T=\mathbb{I}_{n \times n}.
    \end{aligned}
\end{eqs}
In the rotor case, according to the analysis in Sec.~\ref{sec:Clifford_group_structure}, $A$ should be in $\text{GL} (n,\mathbb{Z})$ and the upper-right block $B$ must be empty. Hence, rotor passive transformations are written as
\begin{eqs}
    S=\begin{pmatrix}
        A & 0\\
        0 & A
    \end{pmatrix},~~~ \text{where}~~~ \begin{aligned}
        &AA^T=\mathbb{I}_{n \times n},\\
        &A \in \text{GL}_n (\mathbb{Z}).
        \end{aligned}
\end{eqs}

In words, to preserve the square of the total angular momentum, $\sum_i l_i^2$, a passive transformation should be an element of $\text{O}(n)$. 
In the meantime, $l_i$'s are integers, so the transformation should take value in $\text{GL}_n (\mathbb{Z})$.
These conditions fix the rotor passive transformation group to be the signed symmetric group (\textit{a.k.a.}\ the hyperoctahedral group),
\begin{eqs}\label{eq:rotor_passive}
    \text{O}(n)\cap\text{GL}_n (\mathbb{Z})=\mathbb{Z}_2\wr S_n~,
\end{eqs}
generated by \textsc{swap} and parity flip $\textsc{p}$ (where \(\wr\) is the wreath product).

While the oscillator group is a rich and complex amalgamation of \(\mathsf{SU}(2)\) mode-mixing transformations (beam splitters) and \(\text{U}(1)\) single-mode rotations (phase shifters) \cite{serafini2017quantum}, its rotor counterpart consists of only permutations and momentum parity flips.
A parity flip is inherited from the \(\pi\)-phase shift, while a \textsc{swap} is passed down from the 50/50 beam splitter.

In passing, we note that the lack of beam-splitters and squeezing precludes us from mapping any interesting Gaussian bosonic channels \cite{lloyd2012gaussian} into the rotor state space.

\section{Rotor ``Gaussian'' states}\label{sec:rotor_Gaussian_state}

Gaussian states play a fundamental role in many aspects of quantum physics of harmonic oscillator systems. They are a class of states which can be described just by their first and second moments. The evolution of bosonic Gaussian states under symplectic unitaries can be efficiently simulated by classical algorithms via tracking the changes of their first and second moments \cite{bartlett2002efficient,mari2012positive}.
In the cases of planar rotor, the symplectic representation of Clifford circuits also gives us a classically efficient simulation algorithm \cite{bermejo2016normalizer}.

The nullifier states, coherent states, squeezed states, as well as thermal states of free Hamiltonians, are all important examples of oscillator Gaussian states.
We wish to enumerate states of the rotor system with similar properties to oscillator Gaussian states.

In this section, we discuss the analogy of nullifier states and coherent states in rotor systems and compare them to their counterparts in oscillator systems.
We also discuss the transformation of Josephson-junction rotor Hamiltonians under the rotor Clifford group.

\subsection{Rotor nullifier states}\label{sec:rotor_nullifier_states}

For oscillator systems, the position and momentum eigenstates are called \textit{nullifier states} \cite{gu2009quantum}.
They are $\delta$-functions localized at given position or momentum values and can thought of as infinitely squeezed coherent states.
Oscillator nullifier states are an important class of oscillator Gaussian states and have been studied in the context of continuous-variable quantum computing and error correction \cite{gu2009quantum,menicucci2011graphical,chen2014experimental,vuillot2019toric,asavanant2019generation,xu2023qubit}.

\textcolor{black}{Single-rotor nullifier states are the angle phase states and the angular momentum eigenstates. However, only rotor angular momentum eigenstates are proper Gaussian states for rotors while the phase states are not Gaussian \cite{rigas2010non}.
One way to understand this is to project oscillator Gaussian states into the embedded rotor subspace and look for those states that remain Gaussian after projection. 
Only momentum states satisfy this constraint, as rotor position states --- an infinite superposition of periodically identified oscillator position states --- are no longer Gaussian.}

A multi-rotor nullifier state is defined by a vector \(\Vec{\boldsymbol{l}}=(l_1,...,l_n)^T\) denoting the momentum of each rotor,
\begin{eqs}
\ket{\Vec{\boldsymbol{l}}}=\bigotimes_{j=1}^{n}\ket{\hat{l}_{j}=l_{j}}\ensuremath{}~.
\end{eqs}
We show that rotor nullifier states are closed under the Clifford operation.

If we apply a rotor Clifford circuit $U$ on $\ket{\Vec{\boldsymbol{l}}}$, the evolution of the eigenoperator is
\begin{eqs}\label{eq:angular_momentum_heisenberg}
    \hat{\mathbf{l}} \rightarrow U \hat{\mathbf{l}} U^\dagger.
\end{eqs}
This can be represented by multiplying symplectic matrices on the quadrature vector. Importantly, we only need to track the changes of angular momentum operators since the rotor Clifford group does not mix in (continuous) positions into (integer) momenta.
 In other words, since the group has a semi-direct product structure, we can always write an element as $g=g_Hg_N$, where $g_H \in H$, $g_N \in N$, and Eq.~\eqref{eq:angular_momentum_heisenberg} becomes
\begin{eqs}\label{eq:mom-rot}
    &U  \hat{\mathbf{l}} U^\dagger= g_H g_N  \hat{\mathbf{l}}
    = g_H  \hat{\mathbf{l}}
    = A_U^{-1} \hat{\mathbf{l}}.
\end{eqs}
Here we use $g_N \hat{\mathbf{l}}=\hat{\mathbf{l}}$, because \textsc{quad} and \textsc{cphs} gates (block off-diagonal gates) all commute with $\hat{\mathbf{l}}$.

Re-expressing the above in the Schrodinger picture, any nullifier state \(\ket{\Vec{\boldsymbol{l}}}\) --- an eigenstate of \(\hat{\mathbf{l}}\) with eigenvalues \(\Vec{\boldsymbol{l}}\) --- transforms into the state \(U \ket{\Vec{\boldsymbol{l}}}\) --- an eigenstate of \(A_U^{-1} \hat{\mathbf{l}}\) with eigenvalues \(A_U^{-1}\Vec{\boldsymbol{l}}\).
This new state is still a tensor product of angular momentum eigenstates. Therefore, rotor nullifier states are closed under the action of rotor Clifford group and this evolution is fully captured by symplectic transformations.

\subsection{Coherent states}\label{sec:rotor_coherent_state}

In this subsection, we review rotor coherent states \cite{kowalski1996coherent} and show that they arise from projecting oscillator coherent states into the embedded rotor subspace.
Per the discussion from Sec.~\ref{sec:non_exist_squeezing}, squeezing such coherent states cannot be implemented unitarily.
However, we show that squeezed rotor states previously introduced in Ref.~\cite{kowalski2002uncertainty} are equivalent to applying an adjustable regularizer $e^{-\Delta \hat{l}^2/2}$ to a fiducial coherent state.
This ``regularization'' is similar to the manner of realizing finite-energy GKP states \cite{menicucci2014fault}.

Oscillator coherent states are (right) eigenstates of the oscillator annihilation operator. A class of rotor coherent states \cite{kowalski1996coherent} can be defined analogously by the equation
\begin{eqs}
    e^{i \hat{a}} \ket{\xi}=e^{i(\hat{\theta}+i\hat{l})} \ket{\xi}.
\end{eqs}
Above, $\hat a$ can be thought of as an effective lowering operator for the rotor.
It has to be in the exponent because $\hat{\theta}$ is not well-defined as a standalone operator. Using the commutation relation between $\hat{\theta}$ and $\hat{l}$,
\begin{eqs}
    e^{i \hat{a}} =e^{i \hat{\theta}} e^{-\hat{l}-\frac{1}{2}}=X(1) e^{-\hat{l}-\frac{1}{2}}.
\end{eqs}

Expressing rotor coherent states in the momentum basis yields
\begin{eqs}\label{eq:rotor_coherent_eigenstate}
    &\ket{\xi} =\sum_{l \in \mathbb{Z}} \xi^{-l} e^{-l^2/2} \ket{l},\\
    &e^{i \hat{a}} \ket{\xi} =\xi \ket{\xi}, \quad \xi \in \mathbb{C}-\{0\},\\
    &e^{- (\ln{\chi}) \hat{l}} \ket{\xi} =\ket{\chi \xi},\quad \chi>0.
\end{eqs}
Note that $\ket{\xi=0}$ is not allowed because its wavefunction diverges.
In this form, we can see that rotor coherent-state coefficients are evaluations of the momentum Gaussian wavefunction of the oscillator coherent state at integer momenta.
Contrary to the oscillator case, the (rotor) Wigner functions of coherent states $\ket{\xi}$ have negative parts~\cite{rigas2010non}.

One can define a rotor displacement operator $D(\alpha)$ as
\begin{eqs} \label{eq:rotor_displacement}
    D(\alpha)&= \exp (\frac{\alpha}{2} \hat{a}^\dagger- \frac{\alpha^*}{2} \hat{a})= \exp (i(d \hat{\theta}-c \hat{l})), \\
    &= e^{-icd/2} X(d) Z(-c)\\
    & \text{where } \alpha=c+id, \quad c\in \mathbb{T}, ~ d \in \mathbb{Z}.
\end{eqs}
Because of the discrete nature of the angular momentum basis, we can only apply discrete displacement along the angular momentum direction.
Displacement operators are an alternative way to express the rotor Pauli group \(\mathcal{P}^{\text{rot}}_n\).
The difference between the displacement operators of rotors and oscillators is that $\alpha\in\mathbb{C}$ for oscillator systems, whereas $\text{Im}(\alpha)\in \mathbb{Z}$ in rotors.

Conjugating $e^{i \hat{a}}$ by $D(\alpha)$ performs a displacement transformation on $e^{ i \hat{a}}$:
\begin{eqs}
D(\alpha)e^{i\hat{a}}D(\alpha)^{\dagger}&=X(d)Z(-c)e^{i\hat{a}}Z(-c)^{\dagger}X(d)^{\dagger}\\&=e^{i(\hat{a}-\alpha)}.
\end{eqs}
When the displacement operator $D(\alpha)$ acts on a rotor coherent state $\ket{\xi}$, we have 
\begin{eqs}\label{eq:coherent-disp}
    D(\alpha) \ket{\xi} = (\xi e^{\alpha/2})^d\ket{e^{i \alpha} \xi},
\end{eqs}
giving rise to another coherent state, up to a constant factor. 
In contrast to oscillator coherent states, here the displaced state has to be renormalized relative to the initial state when \(d\) is nonzero.

As discussed in Sec.~\ref{sec:passive_transformation_group}, the transformations that leave the direct product structure of rotor coherent states invariant form the group \textcolor{black}{$\mathcal{P}^{\text{rot}}_n\rtimes(\mathbb{Z}_2\wr S_n)$}. The action of the Pauli part $\mathcal{P}_n^{\text{rot}}$ is shown in Eq.~\eqref{eq:coherent-disp}.
The action of the permutation group $S_n$ swaps the order of the rotors in the direct product sequence.
The action of the $\mathbb{Z}_2$ part, generated by parity $\textsc{p}$, acts on a rotor coherent state as
\begin{eqs}
    \textsc{p}\ket{\xi}=\ket{\xi^{-1}}.
\end{eqs}

Rotor coherent states cannot be used to approximate rotor nullifier states by applying a squeezing operator, because squeezing is not a unitary operation for rotor systems (see Sec.~\ref{sec:non_exist_squeezing}).
However, one can still introduce one more parameter to the rotor coherent state to ``simulate'' squeezing,
\begin{eqs}
    \ket{\xi=1}_\Delta \propto E_\Delta \ket{\theta=0}=\sum_{l \in \mathbb{Z}} e^{-\frac{\Delta l^2}{2}} \ket{l},
\end{eqs}
where $E_\Delta(\hat{l})=\exp(-\Delta \hat{l}^2/2)$ is called a regularizer, and \(\Delta\in[0,\infty)\) is a regularization parameter.

The above states simulate squeezing by smoothly interpolating between states of fixed position and momentum. If we take $\Delta \rightarrow 0$, then the ``squeezed'' rotor coherent state approaches the phase state as $\ket{\xi=1}_{\Delta\rightarrow 0}\rightarrow \ket{\theta=0}$. If we take $\Delta \rightarrow \infty$, then the ``squeezed '' rotor coherent state approaches the angular momentum state  $\ket{\xi=1}_{\Delta\rightarrow \infty}\rightarrow \ket{l=0}$.

Following the finite-energy approach in Refs.~\cite{menicucci2014fault,duivenvoorden2017single,royer2020stabilization}, we define the finite-energy version of Pauli operators via
\begin{eqs}
    X(m)_\Delta&= E_{\Delta}(\hat{l})  X(m) E_{\Delta}^{-1}(\hat{l})=X(m) E_\Delta (\hat{l}+m) E_\Delta^{-1}(\hat{l}),\\
    Z(\phi)_\Delta &= E_{\Delta}(\hat{l})  Z(\phi) E_{\Delta}^{-1}(\hat{l})= Z(\phi).
\end{eqs}
Though the finite-energy version of operators are non-unitary, they follow the same Pauli algebra as regular Pauli operators.
The Pauli $Z$ component is unaffected by the energy regularizer. We only need to apply regulator in $l$ bases as $\exp(-\Delta \hat{l}^2/2)$ because the phase states are un-normalizable while angular momentum states are normalized. This regularizer can also be regarded as an unnormalized thermal state density operator whose Hamiltonian is quadratic in angular momentum.

\subsection{Clifford-group orbits of the Josephson junction}\label{sec:josephson_junction}

In this subsection, we discuss how the Josephson junction Hamiltonian \cite{girvin2014circuit,blais2021circuit} transforms under the action of the rotor Clifford group (\textit{without} performing the cosine expansion and approximating all rotors as oscillators).
Rotor Clifford operations perform a unique mixture of oscillator-like Gaussian manipulations and qubit-like Clifford conjugations.

The Josephson junction allows for tunneling of superconducting paired electrons called Cooper pairs between the two islands that the junction connects.
The planar rotor angular momentum $\hat{l}$ can be associated with the difference of Cooper pairs $\hat{l}=\hat{n}_R-\hat{n}_L$ across the junction, where $\hat{n}_{L/R}$ is the Cooper-pair number operator of the left and right island.
In such a case, the parity flip transformation \(\textsc{p}\) corresponds to swapping the definition of the two islands.

The Josephson junction Hamiltonian is
\begin{eqs}
    H=-2E_J \cos{\hat{\theta}}-E_C \hat{l}^2=-E_J (e^{i\hat{\theta}}+ e^{-i\hat{\theta}})-E_C \hat{l}^2,
\end{eqs}
which includes a quadratic kinematic term $E_C \hat{l}^2$ and a periodic potential $E_J \cos \hat{\theta}$.
If we consider this Hamiltonian in the picture of embedded rotor where $\hat{\theta} \rightarrow \hat{q}, \hat{l} \rightarrow \hat{p}$, it is equivalent to a boson with quadratic dispersion relation $\hat{p}^2$ that moves in a 1d periodic potential $\cos \hat{q}$.
This can be regarded as a GKP Hamiltonian \cite{gottesman2001encoding,rymarz2021hardware,kolesnikow2023gottesman} with only $\theta$ variable periodic while leaving $l$ variable non-compact.

Suppose we have $n$ decoupled Josephson junctions, where the total Hamiltonian is written as
\begin{eqs}\label{eq:JJ_hamiltonian}
    H_{\text{tot}}=& \sum_{j=1}^n \Big(-E_J (e^{i\hat{\theta}_j}+ e^{-i\hat{\theta}_j})-E_C \hat{l}_j^2 \Big),\\
    =& \sum_{j=1}^n -E_J (X_j(1)+ X_j(1)^\dagger)-E_C \Vec{\mathbf{l}}^T \Vec{\mathbf{l}}.
\end{eqs}
To see how the above Hamiltonian transforms under a rotor Clifford transformation, we recall that in Eq.~\eqref{eq:symplectic_coefficient} how the symplectic matrix acts on the coefficient vector. Instead of directly transforming the momentum-position coefficients in a Pauli string, we use the Heisenberg picture and act on the position-momentum \textit{operators}.
We abuse the notation and suppose $\hat\theta$ could be defined without exponentiation and write the collection of them as a vector $\hat{\boldsymbol{\theta}}=(\hat{\theta}_1,...,\hat{\theta}_n)^T$.
The symplectic matrix acting on the operators can be written as
\begin{eqs}\label{eq:Clifford_operator_change}    (\hat{\boldsymbol{\theta}}^T|\hat{\mathbf{l}}^T)\begin{pmatrix}
        A & 0\\
        CA & (A^T)^{-1}
    \end{pmatrix}=((A^T\hat{\boldsymbol{\theta}})^T+(A^TC\hat{\mathbf{l}})^T|(A^{-1}\hat{\mathbf{l}})^T).
\end{eqs}

The position-operator part of the Hamiltonian becomes mixture of position and momentum Paulis,
\begin{eqs}
    \sum_{j=1}^n (e^{i\hat{\theta}_j} +\text{h.c.}) \rightarrow \sum_{j=1}^{n} (e^{i \sum_{k=1}^n A_{kj} \hat{\theta}_k}e^{i\sum_{k=1}^n(A^TC)_{jk}\hat{l}_k}+\text{h.c.}).
\end{eqs}
The $Z$ part of Pauli operators comes in because of the off diagonal matrix $Q_{XZ}=CA$, corresponding to $\textsc{cphs}$ and $\textsc{quad}$ gates.
If we only applied an element of the CSS subgroup $H$ from Eq.~\eqref{eq:clifford_generator}, then the original $X_j(1)$ only transforms into a product of Pauli \(X\) operators of rotors.

On the other hand, the angular-momentum part of the Hamiltonian always transforms as
\begin{eqs}
    \hat{\mathbf{l}}^T \hat{\mathbf{l}}\rightarrow \hat{\mathbf{l}}^T (A)^{-1} (A^T)^{-1} \hat{\mathbf{l}},
\end{eqs}
remaining a quadratic form after any symplectic transformation.

The difference between the transformation behaviours of the $\theta$ and $l$ variables illustrates the difference between periodic bounded variables and unbounded variables. For an unbounded variable $l$, the quadratic term in the Hamiltonian remains a quadratic term. This is the same as in the oscillator system where a quadratic Hamiltonian remains quadratic after Gaussian transformation.

However, for periodic variable $\theta$, if we start from a trivial Hamiltonian $H=\sum_{j=1}^n X_j$ and apply a Clifford circuit, we get a stabilizer Hamiltonian $\hat{H}'=\sum_{j=1}^n \hat{S}_j$, where $\hat{S}_j$s are Pauli strings of multiple rotors. This behaviour resembles the Clifford transformation of Pauli operators for qudits. This is due to the exponential required to properly express the periodic rotor position.

\section{Homological rotor codes}\label{sec:homo_rotor_revisit}

In Ref.~\cite{vuillot2023homological}, the authors defined a construction of rotor CSS codes and related the underlying structure of the codes to homology. 
For planar rotor systems including superconducting qubits, diatomic molecules, and nano-rotors, this construction provides a systematic way to construct CSS-like quantum codes that encode both discrete-variable (logical qudits) and continuous-variable (logical rotors) subspaces.
Here, we revisit these homological rotor codes.

Homological rotor codes are stabilizer codes \eczoo{stabilizer}, meaning that the codespace is the \(+1\)-eigenvalue eigenspace of a group of mutually commuting rotor Pauli strings.
Furthermore, these codes are of CSS type \eczoo{css}, which means that their stabilizer group is generated by purely $X$ and purely $Z$ Pauli strings.

The stabilizer generators of an \(n\)-rotor code are described by matrices $H_X$ and $H_Z$ of integer-valued entries with size $r_X\times n$ and $r_Z\times n$, respectively. The stabilizer group is
\begin{eqs}\label{eq:rotor_stabilizer_group}
    \mathcal{S}=\{X(\Vec{\boldsymbol{m}}^T H_X)Z(\Vec{\boldsymbol{\phi}}^T H_Z),~\forall \Vec{\boldsymbol{m}}\in\mathbb{Z}^{r_X},~\forall \Vec{\boldsymbol{\phi}}\in\mathbb{T}^{r_Z}\}.
\end{eqs}
Mutual commutation imposes the condition
\begin{eqs}
    &Z(\Vec{\boldsymbol{\phi}}^T H_Z)X(\Vec{\boldsymbol{m}}^T H_X)\\=&\exp{(i\Vec{\boldsymbol{\phi}}^T H_Z H_X^T\Vec{\boldsymbol{m}})}X(\Vec{\boldsymbol{m}}^T H_X)Z(\Vec{\boldsymbol{\phi}}^T H_Z)\\
    =&X(\Vec{\boldsymbol{m}}^T H_X)Z(\Vec{\boldsymbol{\phi}}^T H_Z)\implies H_Z H_X^T=0.\\
\end{eqs}

The matrices $H_X$ and $H_Z$ can also be viewed as maps $\partial=H_X:\mathbb{Z}^{r_X}\to\mathbb{Z}^n$ and $\sigma=H_Z^T:\mathbb{Z}^n\to\mathbb{Z}^{r_Z}$. The above condition implies that composition of the corresponding maps is zero, $\sigma\circ\partial=H_XH_Z^T=0$.
This enables us to define a chain complex over the integers, with $\partial$ and $\sigma$ as its boundary maps.

Code properties can be equivalently stated in terms of properties of the chain complex.
The logical space is given by the complex's first homology group,
\begin{eqs}
\text{H}_1(\mathbb{Z})=\ker(H_Z^T)/\text{im}(H_X)~,
\end{eqs}
which is the quotient group formed by cosets of elements of the image of \(H_X\) in the kernel of \(H_Z^T\).

Reference \cite{vuillot2023homological} showed that the above homology group is a product of integer factors \(\mathbb Z\) --- each denoting a logical rotor --- and discrete factors \(\mathbb{Z}_d\) --- each denoting a \(d\)-dimensional logical qudit.
The latter pieces yield finite-dimensional codespaces and come from what is known as \textit{torsion}.
This effect is not present in analogously defined oscillator and Galois-qudit CSS codes \eczoo{css}.

The impact of torsion is present already in the case of a single rotor, where one obtains a finite-dimensional codespace that is related to both rotor and oscillator GKP codes.

\begin{example}\label{example:torison_GKP}
For a single rotor homological code, there is only one stabilizer, $X(dN)=e^{idN\hat{\theta}}$ with a one-by-one matrix \(H_X = dN\).
The other matrix \(H_Z=0\) since no other integer commutes with \(H_X\).
The codespace is finite dimensional, with the $dN$ codewords given by
\begin{eqs}\label{eq:embedded_oscillator_GKP}\ket{j}=\sum_{m\in\mathbb{Z}}\ket{l=mdN+j},~\forall j\in\mathbb{Z}_{dN}.
\end{eqs}

We can relate the above homological rotor code to the ordinary oscillator GKP codes [see Fig.~\ref{fig:oscillator_GKP_torsion}(a)].
Embedding this rotor in an oscillator by thinking of the rotor momentum states \(|l\rangle\) as oscillator momentum states, we see that the codewords correspond to those of the GKP codes. And the stabilizer for the embedded rotor together with the $H_X$ stabilizer exactly correspond to the pair of the stabilizers of the GKP codes.

Alternatively, if we restrict the exponent of $Z$-type stabilizers from $\phi\in\mathbb{T}$ to some subgroup $\phi\in\mathbb{Z}_N$, we can observe that the \(Z\)-type Pauli \(\hat{S}_Z=Z(\frac{2\pi}{N})\) does commute with \(X(dN)\).
Treating this as an additional stabilizer yields the rotor GKP codes, and adding this stabilizer can be thought of as concatenating the above homological rotor code with modular-qudit GKP codes \eczoo{qudit_gkp} [see Fig.~\ref{fig:oscillator_GKP_torsion}(b)].

We discuss both of the above relations in Sec.~\ref{sec:rotor_gkp}.

\end{example}

\begin{figure}[t]
    \centering
    \includegraphics[width=\columnwidth]{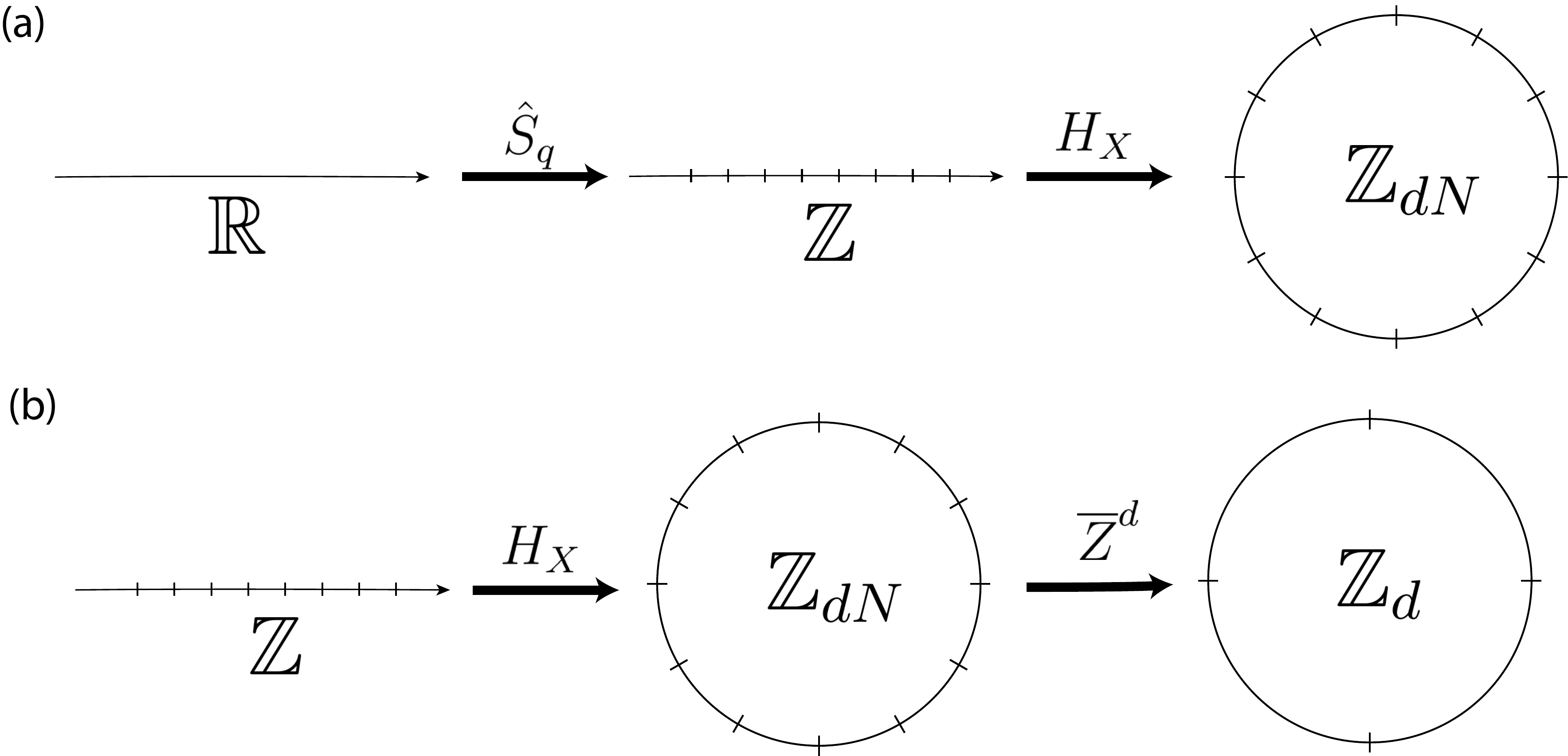}
    \caption{(a) An oscillator GKP code that encodes a $\mathbb{Z}_{dN}$ qudit is a homological embedded rotor code with $H_X=dN$ and torsion part $\mathbb{Z}_{dN}$. (b) A rotor GKP code can be viewed as a concatenation of a homological rotor code with $H_X=dN$ and a modular-qudit GKP code. We discuss both of these relations in Example \ref{example:torison_GKP} and Sec.~\ref{sec:rotor_gkp}.}
    \label{fig:oscillator_GKP_torsion}
\end{figure}

\subsection{Torsion and Smith normal form}\label{sec:torsion}
To build up intuition behind torsion, we calculate the homology group by studying the kernels and images in the Hilbert space of rotors which form a lattice $\mathbb{Z}^n$.

To calculate the homology, we first need to identify the sub-lattice $\ker(H_Z^T)$.
We argue that $H_Z$ only contributes to defining the sub-lattice $\ker(H_Z^T)$ but does not contribute to the torsion part. This is because their stabilizer set is a continuum of stabilizers \(Z(\Vec{\boldsymbol{\phi}}H_Z)\), indexed by the real-valued vectors \(\Vec{\boldsymbol{\phi}}\).
So as long as the true matrices \(H_Z\) are integer-valued, their rescaled versions yield the same stabilizer group. Therefore, we can effectively rescale each row of $H_Z$ to obtain a vector with the smallest possible lattice spacing.

We then need to calculate $\text{im}(H_X)$. This is different from the $H_Z$ case because the stabilizers are $X(\Vec{m}H_X),~\forall \Vec{m}\in \mathbb{Z}^{r_X}$. As the coefficients are discrete, we are not allowed to arbitrarily rescale row vectors in $H_X$.
These vectors generate a lattice whose spacing is determined by their length. If a row vector $\Vec{v}_x$ is $k$ times the unit vector in this direction, then the image of $\Vec{v}_x$ under coefficient $\mathbb{Z}$ will skip the grid points in between the starting and ending points of $\Vec{v}_x$. When we take the quotient of $\text{im}(H_X)$, the direction of $\Vec{v}_x$ becomes a circle with $k$ grid points, giving a $k$ dimensional qudit and contributing a $\mathbb{Z}_k$ factor to the torsion part.

The systematic method of identifying the elements in $\text{im}(H_X)$ that do not have unit length in one direction but are not multiples of other vectors is to calculate the Smith normal form. For an $r_X\times n$ integer-valued matrix $H_X$, its \textit{Smith normal form} \cite{dumas2003computing,web:smith_normal} $D$ is given by
\[U H_X V= D,\] where  $D$ is in the diagonal form, meaning $D_{ij}=d_i \delta_{ij},~1\leq i\leq r_X$. Here $d_i$ are positive integers and satisfy the condition that each $d_i$ is a divisor of $d_{i+1}$, for $1\leq i<r_X$. $U$ is an $r_X\times r_X$ unimodular matrix and $V$ is an $n\times n$ unimodular matrix. The diagonal entries of $D$, $\{d_1,...,d_{r_X}\}$, are unique and they are called the invariant factors.

The Smith normal form is closely related to the homology group defined by $H_Z$ and $H_X$. To be concrete, let the diagonal element of $D$ be $d_1, ..., d_{r_X}$, then the homology group defined by $H_Z^T$ and $H_X$ with $H_X H_Z^T=0$ can be obtained as
\begin{equation}
\ker(H_Z^T)/\text{im}(H_X)=\Big(\bigoplus_{i=1}^{r_X} \mathbb{Z}_{d_i} \Big)\oplus \mathbb{Z}^{n-r_X-r_Z}.
\end{equation}
The $\bigoplus_{i=1}^{r_X} \mathbb{Z}_{d_i}$ is the torsion part and the $\mathbb{Z}^{n-r_X-r_Z}$ is the free part. If $d_i=1$, then $\mathbb{Z}_{d_i}=1$, which is trivial.

The meaning of the matrix $U$ in $UH_XV=D$ is performing linear combinations of the generators of the $X$ stabilizer group. It has no physical consequences to the logical space of the code. The physical meaning of $V$ is actually the \textit{decoding circuit} of the code. To see this, recall that a Clifford transformation that preserves the CSS structure of the code should not mix the $l$ and $\theta$ quadratures [see discussions around Eq.~\eqref{eq:Clifford_operator_change}], it should be only an element of the CSS Clifford subgroup $H$  from Eq.~\eqref{eq:clifford_generator}. The $H_X$ and $H_Z$ matrices collect the phase and angular momentum operators so they should change as Eq.~\eqref{eq:Clifford_operator_change} under the CSS Clifford subgroup $H$,
\begin{eqs}
    H_X^\prime=H_X A\quad\quad\text{and}\quad\quad H_Z^\prime=H_Z (A^{-1})^T.
\end{eqs}
Indeed, the CSS condition is preserved as $H_X^\prime H_Z^{\prime T}=0$. So $V$ represents a special decoding Clifford transformation that decouples the entangled stabilizers in $H_X$ to individual $X$ stabilizers on each rotor.

\subsection{Code initialization \& equivalence classes}\label{sec:equivalence_class}

Homological rotor encodings can be performed analogously to those of Gaussian or analog stabilizer codes \cite{lloyd1998analog,gu2009quantum} \eczoo{analog_stabilizer}.
The oscillator encodings consist of a Gaussian transformation applied to an initial \(n\)-mode state, with \(k\) of the modes storing logical information, and \(n-k\) modes initialized in the zero-position state --- the canonical nullifier state. Rotor encodings can be defined analogously using rotor initial states and the rotor Clifford group circuit \(V\), defined in the previous subsection.
A key difference is that obtaining codewords with particular torsion cannot be done by Clifford-group operations and instead requires the initial state on the \(n-k\) rotors to be of a particular form.

Applying elements from the CSS Clifford-subgroup $H$ from Eq.~\eqref{eq:clifford_generator} to a given homological rotor code does not change the torsion structure of the logical subspace. This can be seen from the invariance of the Smith normal form under unimodular transformation.
Consider an $X$-parity check matrix after transformation $H_X'= B H_X A$ where $A$ is an unimodular matrix representing the Clifford transformation, and $B$ is a change of basis of the stabilizer generators. Its Smith normal form $D'$ is given by
\begin{eqs}
    D'=U' H_X' V'=(U B^{-1}) B H_X A (A^{-1} V)= D.
\end{eqs}
So the Smith normal form $D$ is invariant under the action of unimodular matrices $A$ and $B$.

This fact implies that quantum rotor codes with distinct Smith normal forms belong to different classes of rotor codes, and cannot be deformed from one class to another by only applying rotor Clifford gates and changing the basis of the stabilizer group, which is called Clifford deformation.
The equivalence relation between phase space lattices is also studied in the context of harmonic oscillator GKP states \cite{royer2020stabilization,conrad2022gottesman,wu2023optimal,lin2023closest}. Quantum systems with the same Hilbert space and the same number of fundamental degrees of freedom are not necessarily in the same equivalence class. For example, for a $16$ dimensional Hilbert space with two quantum degrees of freedom,  $\mathbb{Z}_2 \bigoplus \mathbb{Z}_8$ and $\mathbb{Z}_4 \bigoplus \mathbb{Z}_4$ are not equivalent via unimodular transformations as they have different Smith normal forms.

Since Clifford deformation never changes the invariant torsion factors $\{d_1,...,d_{r_X}\}$,  the initial state of a homological rotor encoding with a given Smith normal form needs to be a resource state with that form. In the embedding formalism of rotors, these resource states become exactly GKP states, like the code shown in Example \ref{example:torison_GKP}.

Another observation is that the specific form of the Smith normal form (namely, that each $d_i$ is a divisor of $d_{i+1}$) indicates that not any tensor product of qudits with arbitrary dimensions is allowed in the torsion part. Different tensor products may yield the same Smith normal form. Generally, given two qudits with dimensions $a$ and $b$ which are not divisors of one the other, let $\gcd(a,b)=c$. We can write $a=cq$, $b=cp$, so $p,~q$ are coprime. There exist integers $r$ and $s$ such that $pr-qs=1$. For the integer matrix
$C=\begin{pmatrix}
            a & 0 \\
            0 & b
        \end{pmatrix}$,
we can calculate its Smith normal form as
    \begin{eqs}
        JCK= D=\begin{pmatrix}
            c & 0\\
            0 & cqp
        \end{pmatrix},
    \end{eqs}
where
\begin{eqs}\label{eq:unimodularpq}
    J&=\begin{pmatrix}
            1 & 0\\
            -1 & 1
        \end{pmatrix}\begin{pmatrix}
            1 & 1-q\\
            0 & 1
        \end{pmatrix}\begin{pmatrix}
            1 & 0\\
            -s & 1        \end{pmatrix}=\begin{pmatrix}
            pr-s & 1-q\\
            -pr & q
        \end{pmatrix},\\
      K&=\begin{pmatrix}
            1 & 0\\
            r & 1
        \end{pmatrix}\begin{pmatrix}
            1 & pq\\
            0 & 1
        \end{pmatrix}
        =\begin{pmatrix}
            1 & pq\\
            r & pqr+1
        \end{pmatrix}.
\end{eqs}
This means two qudits with dimensions $a=cq$ and $b=cp$ are equivalent to two qudits with dimensions $c$ and $cpq$ respectively, in which  $\gcd(a,b)=c$ and $p$, $q$ are co prime.
We provide a concrete example of merging two logical qudits into a combined qudit, showing that there are different ways to interpret the same encoding.

\begin{example}
    Given an $X$-type parity check matrix $H_X=\begin{pmatrix}
        2 & 0\\
        0 & 3
    \end{pmatrix}$ that corresponds to a composite logical system formed by a $\mathbb{Z}_2$ qubit and a $\mathbb{Z}_3$ qudit, it is equivalent to a $ \mathbb{Z}_6$ qudit via the  sequence of unimodular transformation in Eq.~\eqref{eq:unimodularpq}.
\begin{eqs}\label{eq:example_Smith_normal_form}
        H_X=&\begin{pmatrix}
            2 & 0 \\
            0 & 3
        \end{pmatrix} \rightarrow
        \begin{pmatrix}
            2 & 0 \\
            -2 & 3
        \end{pmatrix} \rightarrow \begin{pmatrix}
            2 & 0\\
            1& 3
        \end{pmatrix} \rightarrow
        \begin{pmatrix}
            1 & -3 \\
            1& 3
        \end{pmatrix}\\
        \rightarrow & \begin{pmatrix}
            1 & 0 \\
            1 & 6
        \end{pmatrix} \rightarrow
        \begin{pmatrix}
            1 & 0 \\
            0 & 6
        \end{pmatrix}=D.
    \end{eqs}
A pictorial description of these transformations is shown in Fig.~\ref{fig:example_Smith_normal_form}.

\end{example}

\begin{figure}[ht]
    \centering
    \includegraphics[width=0.45\textwidth]{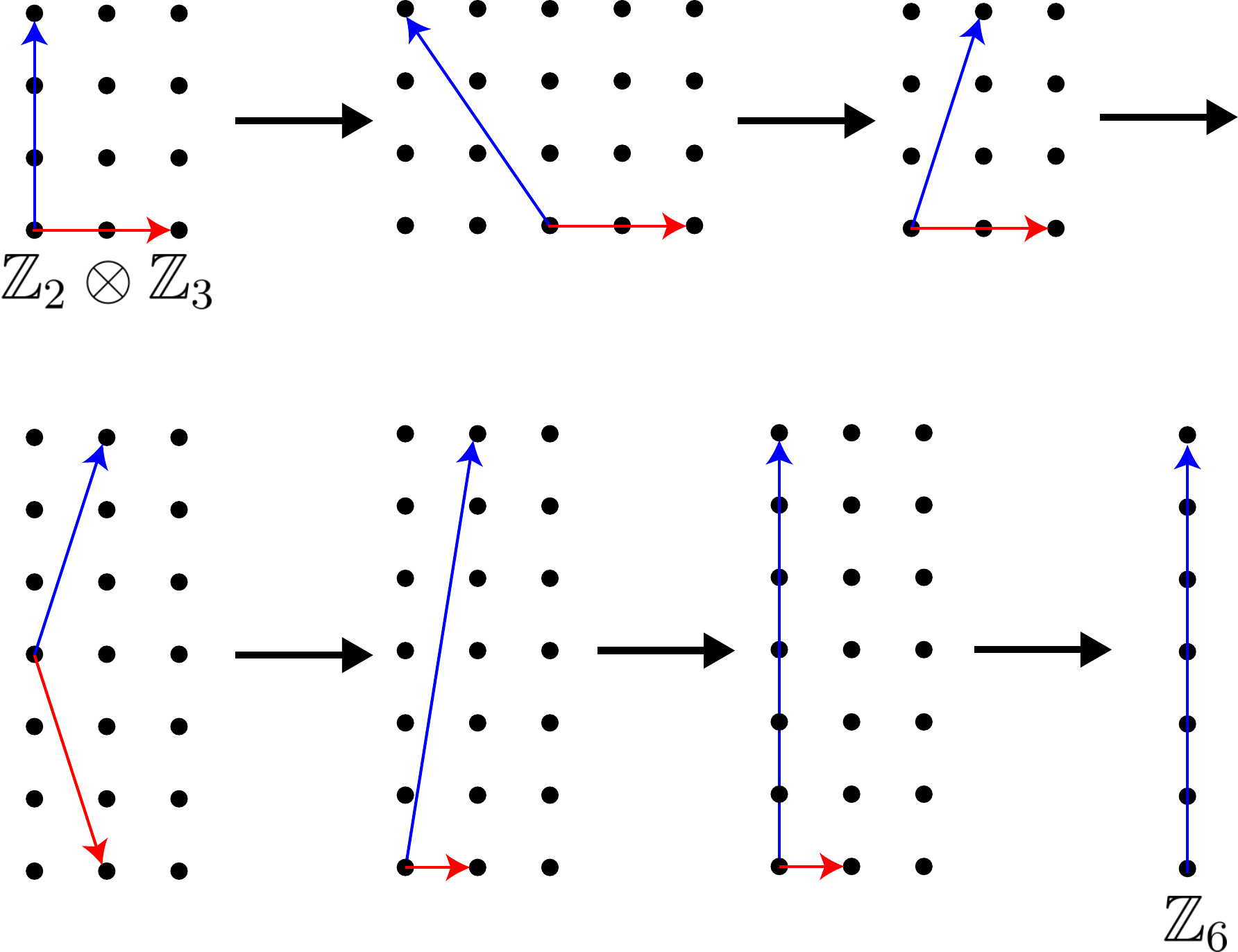}
    \caption{A sequence of unimodular transformation operations by which the isomorphism $\mathbb{Z}_2 \otimes \mathbb{Z}_3 \cong \mathbb{Z}_6$ is realized. }
    \label{fig:example_Smith_normal_form}
\end{figure}

\begin{figure*}[t]
    \centering
    \includegraphics[width=0.95\textwidth]{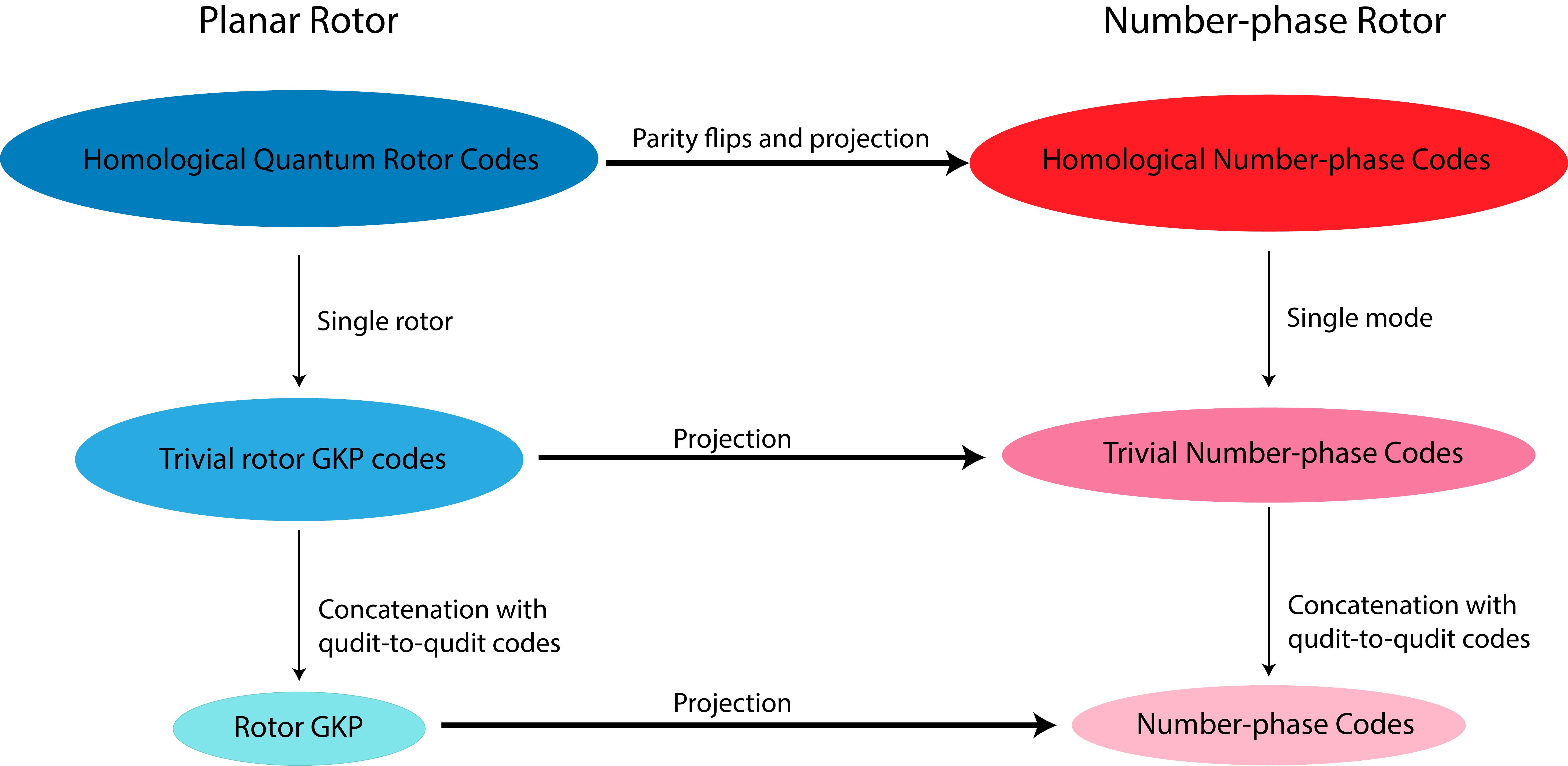}
    \caption{
    Relations between homological rotor codes, trivial rotor GKP codes (oscillator GKP codes from the point of view of the embedded rotor), rotor GKP codes, and homological number-phase codes.
    }
    \label{fig:code_relation}
\end{figure*}

\subsection{Rotor GKP codes are concatenations of homological rotor codes and modular-qudit GKP codes}\label{sec:rotor_gkp}

Analogously to the GKP codes in harmonic oscillators in Eq.~\eqref{eq:embedded_oscillator_GKP}, GKP codes in rotors can be defined by definite a discrete lattice in the rotor phase space, $\mathbb{T}\times \mathbb{Z}$. The rotor GKP code has two stabilizers $\hat{S}_X=X(dN)$, $\hat{S}_Z=Z(\frac{2\pi}{N})$. The codewords are
\begin{eqs}\label{eq:rotor_GKP_codeword}
    \ket{Z_L=\omega^j}_L&=\sum_{m\in \mathbb{Z}} \ket{l=jN+mdN},
\end{eqs}
where $\omega=e^{i\frac{2\pi}{d}}$.
Given the codewords, we have logical Pauli operations
\begin{eqs}\label{eq:rotor_qudit_logical}
    X_L=X(N)=\hat{S}_X^{1/d}, \quad Z_L=Z(\frac{2\pi}{dN})=\hat{S}_Z^{1/d}.
\end{eqs}

Given the similarity between oscillator and rotor GKP codes, we observe a connection between rotor GKP codes and homological rotor codes.

In the framework of homological rotor codes \cite{vuillot2023homological}, the $Z$-type stabilizers are  $Z(\Vec{\boldsymbol{\phi}}^T H_Z)$, $\forall \vec{\boldsymbol{\phi}} \in \mathbb{T}^{r_Z}$.  The phase variables can take continuous values in $\mathbb{T}$. This structure is different from the settings of rotor GKP codes, where the $Z$-type stabilizers are generated by a finite subgroup of $\mathbb{T}$. Despite this difference, we can view the two stabilizers of the rotor GKP code as a two-step concatenation. The two steps of this concatenation are shown in the part (b) of Fig.~\ref{fig:oscillator_GKP_torsion}.

Starting from a single rotor, we first impose the stabilizer $H_X=dN$. This is the code discussed in the previous subsection with codewords and logical operators
\begin{eqs}
&\ket{\overline{Z}=e^{i \frac{2\pi j}{dN}}}= \sum_{m\in \mathbb{Z}} \ket{l=j+mdN},\\
&\overline{X}= X(1)= e^{i\hat{\theta}},\quad \overline{Z}=Z(\frac{2\pi}{dN})=e^{i \frac{2\pi }{dN} \hat{l}}.
\end{eqs}

Then we further impose $\overline{Z}^d=Z(\frac{2\pi}{N})$ as the stabilizer for the second step, thereby concatenating the homological code with a modular-qudit GKP code \cite[Sec.~II]{gottesman2001encoding} that encodes a logical $d$ dimensional qudit into a $dN$ dimensional physical qudit. The effect of the stabilizer $\overline{Z}^d$ is to fix the angular momentum to be a multiple of $N$. Hence, we obtain a $d$-dimensional logical qudit Hilbert space, whose codewords are in Eq.~\eqref{eq:rotor_GKP_codeword} and logical operators are in Eq.~\eqref{eq:rotor_qudit_logical}.

In Appendix.~\ref{app:coherent_Wigner}, we calculate the Wigner function for GKP states in rotor space and find the function does have negativity which corresponds to the ``magic'' for continuous-variable systems. Furthermore, in Appendix~\ref{app:ec_condition_rotor_GKP}, we calculate the error correction condition for regularized rotor GKP states and find it is related to Jacobi $\vartheta$ functions.

\section{Homological number-phase codes}\label{sec:homo_np_rotor_code}

In this section, we turn back to the number-phase rotor introduced in Sec.~\ref{sec:preliminary} and construct oscillator codes correcting photon number changes and dephasing errors. 

\textcolor{black}{Quantum error-correcting codes only protect a certain class of errors. For example, the oscillator GKP codes are designed to correct (linear) quadrature displacements, but they underperform against (radial) dephasing errors.
To protect information from these types of errors, we need to study bosonic codes with rotational symmetry in phase space which can naturally correct both photon number changes and dephasing noise. Therefore, the point of view in Sec.~\ref{sec:preliminary} that viewing an oscillator as a number-phase rotor is useful for this purpose.}

We draw inspiration from the fact that the number-phase oscillator code can be obtained from the rotor GKP code by projecting the former into the non-negative rotor Hilbert subspace. Extending this intuition, we show that \textit{any} homological rotor code can be similarly mapped into number-phase rotors, after flipping the signs of some of the rotor momenta.
This yields multi-mode \textit{homological number-phase codes} for the oscillator, compatible with noise channels where photon loss is present but random-rotation (i.e., dephasing) noise is dominant.

The example below shows what is known for the case of a single mode.

\begin{example}\label{example:np_as_projected_rotor}

The number-phase code is a rotation-symmetric bosonic code protecting against photon loss and dephasing errors \cite{arne2020rotation} (see also \cite{ouyang2021trade}). Its codewords are
\begin{eqs}
    \ket{\overline{0}}_{\text{np}}=\sum_{k\in \mathbb{N}} \ket{2kN},~\ket{\overline{1}}_{\text{np}}=\sum_{k\in \mathbb{N}}  \ket{(2k+1) N}.
\end{eqs}
This code can be obtained by projecting the rotor GKP code from Eq.~\eqref{eq:rotor_GKP_codeword} onto the non-negative rotor Hilbert subspace, and identifying said subspace with the Fock space of the oscillator.

The original rotor GKP code stabilizers, $Z(2k \pi/N)$ and $X(2N)$, become $Z(2k\pi/N)_{\text{np}}$ and $X(2N)^\dagger_{\text{np}}$ after projection, respectively. The new \(Z\)-type operator is still a stabilizer of the number-phase code, but $X(2N)_{\text{np}}$ is no longer a stabilizer since it is not unitary. Nevertheless, the two \(X\) and \(Z\)-type operators form a semigroup, and the codestates are \(+1\)-eigenvalue \textit{right} eigenstates of all semigroup elements \cite{arne2020rotation, albert2022bosonic}.

\end{example}

Encouraged by the above single-rotor example, we can try to take an arbitrary homological rotor code and project its codewords onto the non-negative momentum subspace of each rotor. However, directly this can sometimes result in trivial codewords, as in Example \ref{example:422rotor}.

Nevertheless, there is a simple remedy.
The idea is to flip the orientation of certain rotors so that each codeword has non-trivial support on the non-negative momentum subspace of each rotor.
In the example below, we convert the four-rotor current-mirror code \eczoo{current_mirror} \cite{vuillot2023homological} into its corresponding number-phase rotor code.

\begin{example}\label{example:422rotor}
The four-rotor current-mirror code is defined by parity check matrices

\begin{equation}\label{eq:unfliped_four_rotor}
    H_X=\begin{pmatrix}
        1&-1&0&0\\
        0&0&-1&1\\
        -1&-1&1&1\\
    \end{pmatrix}, \quad
    H_Z=\begin{pmatrix}
        1&1&1&1
    \end{pmatrix}.
\end{equation}
This code encodes a qubit, with logical codewords
\begin{eqs}
&\ket{\bar 0}=\sum_{l_1,l_2,l_3\in\mathbb{Z}}\ket{l_1,2l_2-l_1,l_3.-2l_2-l_3},\\
&\ket{\bar 1}=\sum_{l_1,l_2,l_3\in\mathbb{Z}}\ket{l_1,2l_2+1-l_1,l_3.-2l_2-1-l_3},
\end{eqs}
and logical operators $\overline{X}=X_2(1)^\emph{\dagger} X_4(1)$, $\overline{Z}=Z_3(\pi) Z_4(\pi)$.
Applying the number-phase projection prematurely by cutting off all negative momenta, we see that the only survived state is the trivial state $\ket{0,0,0,0}$. This means that some sign flips are required in order for the projected code to store any logical information.

The required transformation can be done by flipping momenta of the second and third rotors.
The $H_X$ matrix is then
\begin{equation}\label{eq:flipped_four_rotor_stabilizer}
    H_X^+=\begin{pmatrix}
        1&1&0&0\\
        0&0&1&1\\
        0&2&0&2\\
    \end{pmatrix},\quad
    H_Z^+=\begin{pmatrix}
        1&-1&-1&1\\
    \end{pmatrix}~,
\end{equation}
and the code words become
\begin{eqs}\label{eq:flipped_four_rotor_code}
&\ket{\bar 0}_+=\sum_{l_1,l_2,l_3\in\mathbb{Z}}\ket{l_1,2l_2+l_1,l_3,2l_2+l_3},\\
&\ket{\bar 1}_+=\sum_{l_1,l_2,l_3\in\mathbb{Z}}\ket{l_1,2l_2+1+l_1,l_3,2l_2+1+l_3}~.
\end{eqs}
Applying the projector $\Pi_{\geq 0}^{\otimes 4} $ to the codewords yields
\begin{eqs}\label{eq:four_np}
    &\ket{\bar 0}_{\text{np}}=\Pi_{\geq 0}^{\otimes 4
    } \ket{\Bar{0}}_+=\sum_{l_1,l_2,l_3\in\mathbb{N}}\ket{l_1,2l_2+l_1,l_3,2l_2+l_3},\\
&\ket{\bar 1}_{\text{np}}=\Pi_{\geq 0}^{\otimes 4
    }  \ket{\Bar{1}}_+=\sum_{l_1,l_2,l_3\in\mathbb{N}}\ket{l_1,2l_2+1+l_1,l_3,2l_2+1+l_3},
\end{eqs}
with logical operators
\begin{eqs}\label{eq:four_np_logical}
    \overline{X}=X_2(1)^\emph{\dagger} X_4(1)^\emph{\dagger}, \overline{Z}=Z_3(-\pi) Z_4(\pi).
    \end{eqs}
The stabilizer semigroup of Eq.~\eqref{eq:four_np} is generated by
\begin{eqs}\label{eq:four_np_stabilizer}
    \{ &X_1(1)^\emph{\dagger}_{\text{np}} X_2(1)^\emph{\dagger}_{\text{np}},  X_3(1)^\emph{\dagger}_{\text{np}} X_4(1)^\emph{\dagger}_{\text{np}}, X_2(2)^\emph{\dagger}_{\text{np}} X_4(2)^\emph{\dagger}_{\text{np}}, \\
    &Z_1(\phi)_{\text{np}} Z_2(\phi)_{\text{np}}^{\emph{\dagger}} Z_3(\phi)_{\text{np}}^{\emph{\dagger}} Z_4(\phi)_{\text{np}}\}.
\end{eqs}
\end{example}

The following proposition proves that a sequence of flipping operations can always be found to satisfy this constraint.

\begin{proposition}\label{pros:existance_homo_np}
    For a given homological \(n\)-rotor code with matrices $H_X$ and $H_Z$, there exists an unimodular and diagonal matrix $S$, such that all logical codewords of the code defined by $H_X^\prime=H_X S$ and $H_Z^\prime=H_Z (S^T)^{-1}$ have support on the orthant where each rotor has non-negative momentum.
\end{proposition}

\begin{proof}
    The code subspace is a sub-lattice that passes through the origin in $\mathbb{Z}^n$ lattice grids. So it must have support on some certain orthants. Because code words are defined by the quotient under $\text{Im}(H_X)$, each code word is translation invariant in the direction of the row vectors of $H_X$. This means if the $\ker{H_Z^T}$ sub-lattice passes through an orthant, this orthant should overlap with the support of each code word.

    There always exists a series of flipping operations to transform a given orthant to the orthant in which all the integers are positive, denoted as $(+,+,+,...)$. Specifically, an orthant labelled by $(+,-,+,-,-,...)$ is mapped to $(+,+,+,...)$ via the $\mathbb{Z}_2^{n}$ element $I\otimes \textsc{p}\otimes I \otimes \textsc{p} \otimes \textsc{p}...$, with the isomorphism $+\to I$, $-\to \textsc{p}$.
\end{proof}

The procedure to map a homological rotor code to its corresponding homological number-phase code is as follows.
\begin{enumerate}
    \item Given a $n$-rotor homological rotor code with $H_X$ that is not necessarily non-negative, we flip the parities of part of rotors to make $H_X^+=H_X S$ be a non-negative matrix. This transformation will affect $H_Z$ as well, such that $H_Z^+=H_Z S$. Here, $S$ is a diagonal matrix composed by parity flips $\textsc{p}$ whose diagonal elements are $\pm 1$. The effect of $S$ is to convert all $\textsc{cnot}^\dagger$ gates included in the encoding circuit to $\textsc{cnot}$ gates\footnote{It can be shown that
    $\textsc{p}_1 \textsc{cnot}^\dagger_{1\rightarrow 2} \textsc{p}_1=\textsc{p}_2 \textsc{cnot}^\dagger_{1\rightarrow 2} \textsc{p}_2= \textsc{cnot}_{1\rightarrow 2}$.} such that $H_X^+$ is a non-negative matrix. Hence, in practice, we just need to replace every $\textsc{cnot}^\dagger$ by $\textsc{cnot}$ and eliminate pre-existing parity flips $\textsc{p}$ in the encoding circuit.
    \item Project onto non-negative angular momentum subspace by applying $\Pi_{\geq 0}^{\otimes n}$ on each rotor to obtain a homological number-phase code that is stabilized by a stabilizer semigroup given by $H_X^+$ and $H_Z^+$. The resulting homological number-phase code is stabilized by stabilizer semigroup
    \begin{eqs}
    \mathcal{S}^+=\{X(\Vec{\boldsymbol{m}}^T H_X^+)^\dagger_{\text{np}} Z(\Vec{\boldsymbol{\phi}}^T H_Z^+)_{\text{np}},~\forall \Vec{\boldsymbol{m}}\in\mathbb{N}^{r_X},~\forall \Vec{\boldsymbol{\phi}}\in\mathbb{T}^{r_Z}\}.
    \end{eqs}
\end{enumerate}

Similar to the single-mode case, we can assign a distance against rotation errors to the resulting homological number-phase code.
Phase states are not quite orthogonal (see Ref.~\cite{albert2022bosonic} for the case of a single mode),
so homological number-phase codes are not exactly error correcting.
Nevertheless, the sign flipping is done by the single-rotor Clifford operation $\textsc{p}$ from Eq.~\eqref{eq:clifford_generator}, so the entanglement structure of the code, and thus its intended degree of protection, both carry over.
Since configuration-space distances $|(\theta-\phi) \Mod 2\pi|$ are invariant under sign flipping, homological number-phase codes can be described by the distances of their planar-rotor counterparts, but protect against rotation errors in an approximate sense.

\subsection{Number-phase Clifford semigroup encodings}\label{sec:clifford_semigroup_encoding}

Mapping the rotor Clifford group to the number-phase rotor can be done by cutting off all negative momentum states.
The resulting generators for both groups, expressed in terms of the oscillator photon number operator \(\hat n\), are shown in Table~\ref{table:fake_rotor_Clifford}.
We utilize this rotor Clifford semigroup to encode and decode information into homological number-phase codes.

\begin{table}[t]
\centering

\begin{tabular}{l>{\centering}p{01.in}c}
\toprule
Planar rotor & \multicolumn{2}{c}{Number-phase rotor}\tabularnewline
\midrule
Pauli-$X$ \eqref{eq:rotor_pauli_operator} & Photon injection & ${\displaystyle \sum_{n=0}^{\infty}\ket{m+n}\bra{m}}$\tabularnewline
\midrule
Pauli-$Z$  \eqref{eq:rotor_pauli_operator} & Rotation & ${\displaystyle e^{i\varphi\hat{n}}}$\tabularnewline
\midrule
\textsc{cnot} \eqref{eq:rotor_cnot} & Controlled photon injection & ${\displaystyle \sum_{n,m=0}^{\infty}\ket{n}\bra{n}\otimes\ket{m+n}\bra{m}}$\tabularnewline
\midrule
\textsc{quad} \eqref{eq:rotor_quad}~~~~~~~~ & Kerr interaction & ${\displaystyle e^{i\varphi\hat{n}(\hat{n}+1)/2}}$\tabularnewline
\midrule
\textsc{cphs} \eqref{eq:rotor_cphs} & Cross-Kerr interaction & ${\displaystyle e^{i\varphi\hat{n}\otimes\hat{n}}}$\tabularnewline
\bottomrule
\end{tabular}
\caption{Planar-rotor Clifford operations as realized by the ``number-phase rotor'' interpretation of the oscillator from Sec.~\ref{sec:preliminary}; each number-phase operator has the subscript \(\text{np}\) in the main text.
These are obtained by projecting planar-rotor operators into the subspace of non-negative momentum states and interpreting said states as oscillator Fock states \(|n\rangle\).
Parity flip operations are not present because they flip the sign of the momentum states.
Pauli-\(X\) rotations are mapped into powers of the non-unitary Pegg-Barnett phase operator \cite{susskind1964quantum,bergou1991operators,bartlett2002quantum}, which performs photon subtraction and injection.
The \textsc{cnot} gate is mapped to a non-unitary controlled photon injection operator. The Pauli and symplectic groups become semigroups with identity (i.e., monoids), but both non-unitary operators \(U\) remain valid quantum channels because \(U^\dagger U = \mathbb{I}\).
}
\label{table:fake_rotor_Clifford}
\end{table}

After applying the number-phase rotor mapping, we see that the rotor Clifford-group operations $\textsc{quad}_{\text{np},\varphi}$ and $\textsc{cphs}_{\text{np},\varphi}$ are unitary operators, but $\textsc{cnot}_{\text{np}}$ is not.
Indeed, we can show that
\begin{eqs}
    &\textsc{cnot}_{\text{np}}^\dagger \textsc{cnot}_{\text{np}}=\mathbb{I},\\
    &\textsc{cnot}_{\text{np}}\textsc{cnot}_{\text{np}}^\dagger=\sum_{n=0}^\infty \ket{n} \bra{n} \otimes \Pi_{\geq n}.
\end{eqs}
Nevertheless, $\textsc{cnot}_{\text{np}}(\cdot)\textsc{cnot}_{\text{np}}^\dagger$ remains a valid quantum channel that can be utilized for encoding.

Analogous to  Eqs.~\eqref{eq:rotor_cnot}, \eqref{eq:rotor_quad}, and \eqref{eq:rotor_cphs}, we can determine how the number-phase rotor Clifford gates transform the number-phase rotor quadratures:
\begin{eqs}\label{eq:number_phase_rotor_heisenberg}
\textsc{cnot}_{1\rightarrow2,\text{np}}(X(1)_{\text{np}}\otimes\mathbb{I})\textsc{cnot}_{1\rightarrow2,\text{np}}^{\dagger}&=X(1)_{\text{np}}\otimes X(1)_{\text{np}},\\\textsc{cnot}_{1\rightarrow2,\text{np}}(X(-1)_{\text{np}}\otimes\mathbb{I})\textsc{cnot}_{1\rightarrow2,\text{np}}^{\dagger}&=X(-1)_{\text{np}}\otimes X(-1)_{\text{np}},\\\textsc{cnot}_{1\rightarrow2,\text{np}}(\mathbb{I}\otimes Z(\phi)_{\text{np}})\textsc{cnot}_{1\rightarrow2,\text{np}}^{\dagger}&=Z(-\phi)_{\text{np}}\otimes Z(\phi)_{\text{np}},\\\textsc{quad}_{\text{np},\varphi}X(1)_{\text{np}}\textsc{quad}_{\text{np},\varphi}^{\dagger}&=Z(\varphi)_{\text{np}}X(1)_{\text{np}},\\\textsc{cphs}_{\text{np},\varphi}(X(1)_{\text{np}}\otimes\mathbb{I})\textsc{cphs}_{\text{np},\varphi}^{\dagger}&=X(1)_{\text{np}}\otimes Z(\varphi)_{\text{np}},\\\textsc{cphs}_{\text{np},\varphi}(\mathbb{I}\otimes X(1)_{\text{np}})\textsc{cphs}_{\text{np},\varphi}^{\dagger}&=Z(\varphi)_{\text{np}}\otimes X(1)_{\text{np}}.
\end{eqs}
This shows that, despite the presence of non-unitarity, number-phase Clifford semigroup operations can be used to perform conditional operations and extract syndrome information for homological number-phase codes. Notably, the number-phase Clifford semigroup transforms the number-phase Pauli semigroup in the same way as the relevant part of the rotor Clifford group transforms the rotor Pauli operators.
Hence, the symplectic representation we constructed in Sec.~\ref{sec:rotor_Clifford} is applicable to calculate the Clifford transformation of Pauli operators for the number-phase rotor.

\begin{example}\label{example:422_rotor_encoder}
We use the same current-mirror code as in Example \ref{example:422rotor}. The initial state is stabilized by the stabilizer semigroup
\begin{eqs}
    \mathcal{S}_{0}&=\langle X_{1}(1)_{\text{np}}^{\emph{\dagger}},X_{2}(2)_{\text{np}}^{\emph{\dagger}},X_{3}(1)_{\text{np}}^{\emph{\dagger}},Z_{4}(\phi)_{\text{np}}\,~\forall\phi\in\mathbb{T}\rangle\\&=\langle(-1,0,0,0|\boldsymbol{0})^{T},(0,-2,0,0|\boldsymbol{0})^{T},(0,0,-1,0|\boldsymbol{0})^{T},\\&\quad(\boldsymbol{0}|0,0,0,\phi)^{T}\,\,~\forall\phi\in\mathbb{T}\rangle\,,
\end{eqs}
in which we use $\boldsymbol{0}$ to denote $(0,0,0,0)$. The second rotor encodes logical qubit information via a single rotor code with torsion $\mathbb{Z}_2$. We can explicitly write down the initial state as
\begin{eqs}
    \ket{\psi_0}=\sum_{l_1, l_2, l_3\in \mathbb{N}} \ket{l_1}\otimes ( \alpha\ket{2l_2}+\beta \ket{2l_2+1})\otimes \ket{l_3} \otimes \ket{0}.
\end{eqs}
The logical Pauli operators for the initial state are $X=X_2(1)^\emph{\dagger}, Z=Z_2(\pi)$. They admit a vector representation as
\begin{eqs}
     X&=(0, -1, 0, 0 | \boldsymbol{0})^T,~~ Z= (\boldsymbol{0}| 0, \pi, 0, 0)^T.
\end{eqs}

The symplectic representation of the encoding circuit is
\begin{eqs}
    Q^+=\begin{pmatrix}
        A_{\emph{enc}}^+ & 0 \\
        0 & (A_{\emph{enc}}^{+T})^{-1}
    \end{pmatrix}, ~~ A_{\emph{enc}}^+=\begin{pmatrix}
        1 & 1& 0 & 0\\
        0 & 1 & 0 & 1\\
        0 & 0& 1 & 1\\
        0& 0& 0 & 1
    \end{pmatrix},
\end{eqs}
which corresponds to the encoding circuit $U_{\emph{enc}}^+=\emph{\textsc{cnot}}_{1\rightarrow 2,\text{np}} \emph{\textsc{cnot}}_{3\rightarrow 4,\text{np}} \emph{\textsc{cnot}}_{2 \rightarrow 4,\text{np}}$.
The original encoding circuit (before parity flips) of stabilizers shown in Eq.~\eqref{eq:unfliped_four_rotor} is $U_{\emph{enc}}=\emph{\textsc{cnot}}_{1\rightarrow 2}^\emph{\dagger} \emph{\textsc{cnot}}_{3\rightarrow 4}^\emph{\dagger} \emph{\textsc{cnot}}_{2\rightarrow 4}^\emph{\dagger} = \textsc{p}_2 \textsc{p}_3 U_{\emph{enc}}^+ \textsc{p}_2 \textsc{p}_3$.

The stabilizer semigroup and logical operators of those codewords can be calculated according to the method in Sec.~\ref{sec:rotor_Clifford}.
We then transform $\mathcal{S}_0, X, Z$ by multiplying them by $Q^+$, and we obtain
\begin{eqs}
    \mathcal{S}&=\{(-1,-1,0,0|\boldsymbol{0})^T, (0,-2,0,-2|\boldsymbol{0})^T,(0,0,-1,-1|\boldsymbol{0})^T,\\
    &\quad(\boldsymbol{0}| \phi,-\phi,-\phi,\phi)^T\},
    ~ \forall \phi \in \mathbb{T}\\
    \overline{X}&=(0, -1, 0, -1 | \boldsymbol{0})^T,~~ \overline{Z}= (\boldsymbol{0}| -\pi, \pi, 0, 0)^T,
\end{eqs}
which are the stabilizer semigroup and logical operators of the codewords.

After we apply $U_{\emph{enc}}^+$ to $\ket{\psi_0}$, we can see the codestate is written as
\begin{equation}
    \ket{\psi}=U_{\emph{enc}}^+ \ket{\psi_0}=\alpha \ket{\overline{0}}_{\text{np}} + \beta \ket{\overline{1}}_{\text{np}},
\end{equation}
where $\ket{\overline{0}}_{\text{np}}$, $\ket{\overline{1}}_{\text{np}}$ are codewords shown in Eq.~\eqref{eq:four_np}.
\end{example}

The main challenge of syndrome measurement of homological number-phase codes using the Clifford semigroup is measuring $X$-type stabilizers.
This is because the conditional photon subtraction, $\textsc{cnot}_{np}^\dagger$, is not a completely positive trace-preserving (CPTP) map.
However, we can construct the following CPTP map
\begin{eqs}
    \mathcal{D}: \rho &\rightarrow \textsc{cnot}_{\text{np}}^\dagger \rho \textsc{cnot}_{\text{np}}+  \mathcal{P}_C^\dagger \rho \mathcal{P}_C \\
    &= \mathcal{D}_1 (\rho)+ \mathcal{D}_2(\rho)~,
\end{eqs}
where $\mathcal{P}_C= \sum_{n=0}^\infty\ket{n} \bra{n} \otimes (\sum_{m=0}^{n-1} \ket{m} \bra{m}) $ is a projector into the subspace where the photon number on the second mode is smaller than the photon number on the first mode. This CPTP map can decomposed to two non-CPTP maps $\mathcal{D}_1$ and $\mathcal{D}_2$. The $\mathcal{D}_1$ is the map we desire to implement for $X$-type syndrome extraction while the $\mathcal{D}_2$ is undesired. Hence, the syndrome extraction process $\mathcal{D}$ is a probabilistic process. Nevertheless, the syndrome information can be obtained after post-selecting on $\mathcal{D}_1$.

\begin{table*}[t]
    \centering
    \begin{tabular}{llll}
    \toprule
         & planar rotor& embedded rotor& number-phase rotor \\
         \midrule
        physical space & rotors & oscillators& oscillators\\
        \midrule
        logical space & \(k\) rotors & \(k\) rotors& \(k\) oscillators\\
        \midrule
        state on \(n-k\) subsystems~~~~& rotor GKP code & oscillator GKP code & number-phase code\\
        \midrule
        encoding circuit & rotor Clifford & rotor Clifford & Clifford semigroup\\
        \midrule
        error & rotor Pauli \eqref{eq:rotor_pauli_operator} & oscillator displacements~~~ & photon loss and dephasing \eqref{eq:fake_rotor_quadrature}\\
        \bottomrule
    \end{tabular}
    \caption{Comparison between planar rotor GKP-stabilizer, embedded rotor GKP-stabilizer, and number-phase rotor GKP-stabilizer codes outlined in Sec.~\ref{sec:gkp_repetition_rotor}. }
    \label{tab:comparison_GKP_stabilizer}
\end{table*}

\section{GKP-stabilizer codes for rotors}\label{sec:gkp_repetition_rotor}

GKP-stabilizer codes \cite{noh2020encoding} have been proposed as a way to encode logical oscillators into physical oscillators and utilize oscillator GKP states to protect against displacement noise.
In this section, we discuss versions of GKP-stabilizer codes for all three rotors.

In this section, we consider rotor versions of GKP-stabilizer codes \cite{noh2020encoding}, which allow one to protect an infinite-dimensional logical space against displacement noise.
The interpretation of the harmonic oscillator as a number-phase rotor allows us to map such codes back into the oscillator, yielding codes protecting against bosonic dephasing errors.

The encoding of each version consists of placing logical information into \(k\) subsystems (either rotors or oscillators), initializing the remaining \(n-k\) subsystems in a particular resource state, and applying a Clifford circuit.
The noise that the code is suitable for depends on the resource state.
The ingredients for each code are summarized in Table~\ref{tab:comparison_GKP_stabilizer}.

\begin{enumerate}
    \item \textbf{Planar rotor GKP-stabilizer codes} We encode \(k\) logical rotors into \(n\) physical rotors, with \(n-k\) rotors initialized in rotor GKP states. This provides a way to protect logical rotors, which are compact and infinite-dimensional, from physical rotor Pauli $X$ and $Z$ errors.
    Although the homological rotor code \cite{vuillot2023homological} can also encode logical rotors into physical rotors, its $Z$-type stabilizers are continuous, as shown in Eq.~\eqref{eq:rotor_stabilizer_group}, while the rotor GKP-stabilizer code inherits the discrete stabilizer group of the rotor GKP states on the \(n-k\) rotors. \textcolor{black}{This construction can correct the rotor Pauli $X$ and $Z$ error acting on the $k$ logical rotors.}

    \item \textbf{Embedded rotor GKP-stabilizer codes}
    This code can be regarded as a GKP-stabilizer code whose logical space forms rotors instead of oscillators due to the extra embedded-rotor stabilizer constraint placed on the \(k\) logical modes.
    We first encode \(k\) logical rotors into \(k\) logical oscillator modes via the embedded rotor technique.
    Then, we encode the \(k\) logical oscillator modes into \(n\) multiple oscillator modes by initializing \(n-k\) modes in oscillator GKP states\footnote{They are concatenation between oscillator GKP codes and qudit GKP codes (see Sec.~\ref{sec:rotor_gkp}), here we still call them as oscillator GKP codes because of its comb structure.} and applying a Gaussian transformation. We can also regard this code as a concatenation between rotor-to-oscillator code and oscillator GKP-stabilizer code.
    As with the regular GKP-stabilizer code, this code can correct position and momentum displacements.

    \item \textbf{Number-phase rotor GKP-stabilizer codes} \textcolor{black}{Analogous to the planar rotor GKP-stabilizer codes,} we encode \(k\) logical oscillators into \(n\) physical oscillators \footnote{\textcolor{black}{Note that number-phase rotor is another description of the oscillator, hence this construction is for oscillator-to-oscillator codes that correct photon number changing and dephasing errors. }} by initializing \(n-k\) oscillators in number-phase codestates and applying a Clifford semigroup circuit \textcolor{black}{as encoder}. \textcolor{black}{Such construction can be regarded as the polar coordinate generalization of oscillator GKP-stabilizer codes which are formulated in the lattices of Cartesian coordinates (position and momentum). }
    One can also consider encoding the \(k\) rotors in a homological number-phase codeword via the prescription of Sec.~\ref{sec:homo_np_rotor_code}. 
    The number-phase rotor GKP-stabilizer codes can protect logical oscillator modes from quadrature noise of the number-phase rotor, which includes photon loss and dephasing.
\end{enumerate}

In the following of this section, we use the GKP-repetition code as an example to demonstrate all three constructions.

\begin{example}

We would like to encode a single logical rotor into two physical rotors while being able to detect $Z$-type rotor Pauli errors.
We place the logical information, denoted by the function \(\psi(k)\) for integer \(k\), into the first rotor, initialize the second rotor in a rotor GKP state, and apply a \textsc{cnot} gate.
This yields
\begin{eqs}
    \ket{\psi}_{\text{code}}&=\textsc{cnot}_{2\rightarrow1}\left(\sum_{k\in\mathbb{Z}}\psi(k)\ket{k}\right)\otimes\left(\sum_{\ell\in\mathbb{Z}}\ket{m\ell}\right)\\&=\sum_{k,\ell\in\mathbb{Z}}\psi(k)\ket{m\ell+k}\otimes\ket{m\ell}~,
\end{eqs}
where all constituent states are rotor momentum states. This code is stabilized by rotor Pauli strings $X(m) \otimes X(m)$ and $\mathbb{I} \otimes Z(\frac{2\pi}{m})$. \textcolor{black}{The logical Pauli operator of the encoded rotor becomes $\overline{X}(l)=X(l)\otimes \mathbb{I}$, $\forall l\in \mathbb{Z}$  and $\overline{Z}(\phi)= Z(\phi) \otimes Z(-\phi) $, $\forall \phi \in \mathbb{T}$.}

Using Eq.~\eqref{eq:rotor_cnot}, rotor Pauli errors propagate as
\begin{eqs}
\begin{split}
\textsc{cnot}_{1\rightarrow 2} (X(l) \otimes \mathbb{I})&=(X(l) \otimes X(l)) \textsc{cnot}_{1\rightarrow 2}, \\
\textsc{cnot}_{1\rightarrow 2} (\mathbb{I}\otimes Z(\phi))&= (Z(-\phi) \otimes Z(\phi)) \textsc{cnot}_{1\rightarrow 2},
\end{split}
\end{eqs}
enabling us to extract error syndromes as follows.

Suppose we have a codeword corrupted by a dephasing error $Z(\xi_1^Z)\otimes Z(\xi_2^Z)$, the noisy codeword is written as $Z(\xi_1^Z)\otimes Z(\xi_2^Z) \ket{\psi}_{\text{code}}$. Then we initialize an ancillary third mode in the state \begin{eqs}
    \ket{+}_{\text{U}(1)}\propto \sum_{n\in \mathbb{Z}} \ket{nm}=\sum_{n=0}^{m-1}\ket{\theta=\frac{2\pi n}{m}}~,
\end{eqs}
where \(\ket{\theta}\) is a rotor position state.
Finally, we apply $\textsc{cnot}_{3\rightarrow 1} \textsc{cnot}_{3\rightarrow 2}$.
This yields
\begin{widetext}
  \begin{subequations}\label{eq:gkp_repetition_X_syndrome}
\begin{align}
    &\!\!\!\!\!\!\!\!\!\!\!\!\!\!\!\!\!\!\!\!\!\textsc{cnot}_{3\rightarrow1}\textsc{cnot}_{3\rightarrow2}(Z(\xi_{1}^Z)\otimes Z(\xi_{2}^Z))\ket{\psi}_{\text{code}}\otimes\ket{+}_{\text{U}(1)}=\\&=(Z(\xi_{1}^Z)\otimes Z(\xi_{2}^Z)\otimes Z(-\xi_{1}^Z-\xi_{2}^Z))\textsc{cnot}_{3\rightarrow1}\textsc{cnot}_{3\rightarrow2}\ket{\psi}_{\text{code}}\otimes\ket{+}_{\text{U}(1)}\\&=(Z(\xi_{1}^Z)\otimes Z(\xi_{2}^Z)\otimes Z(-\xi_{1}^Z-\xi_{2}^Z))\sum_{k,n,f\in\mathbb{Z}}\psi(k)\ket{m(n+f)+k}\otimes\ket{m(n+f)}\otimes\ket{mf}\\&=(Z(\xi_{1}^Z)\otimes Z(\xi_{2}^Z)\otimes Z(-\xi_{1}^Z-\xi_{2}^Z))\ket{\psi}_{\text{code}}\otimes\ket{+}_{\text{U}(1)}\\&=(Z(\xi_{1}^Z)\otimes Z(\xi_{2}^Z)\otimes Z(-\xi_{1}^Z-\xi_{2}^Z))\ket{\psi}_{\text{code}}\otimes\sum_{f=0}^{m-1}\ket{\theta=\frac{2\pi f}{m}}\\&=(Z(\xi_{1}^Z)\otimes Z(\xi_{2}^Z)\otimes\mathbb{I})\ket{\psi}_{\text{code}}\otimes\sum_{f=0}^{m-1}\ket{\theta=\frac{2\pi f}{m}+\xi_{1}^Z+\xi_{2}^Z}.
    \end{align}
  \end{subequations}
\end{widetext}
Then, by performing a phase-basis projective measurement on the third ancillary mode, we extract the syndrome $(\xi_1^Z+\xi_2^Z) \Mod \frac{2\pi}{m}$. \textcolor{black}{The correctable $Z$-error would be in the interval $(\xi_1^Z+\xi_2^Z) \in [-\frac{\pi}{m},\frac{\pi}{m})$ where $(\xi_1^Z+\xi_2^Z) \Mod \frac{2\pi}{m}=\xi_1^Z+\xi_2^Z$. Then we apply the decoding unitary $\textsc{cnot}_{2\rightarrow1} ^\dagger$ to the noisy codeword and obtain 
\begin{eqs}
    (Z(\xi_1^Z+\xi_2^Z) \otimes \mathbb{I}) \Big( \sum_{k \in \mathbb{Z}} \psi(k) \ket{k}\Big) \otimes \Big( \sum_{\mathcal{l} \in \mathbb{Z}} \ket{ml}\Big).
\end{eqs}
Then we perform error correction by applying $Z(-\frac{\xi_1^Z+\xi_2^Z}{2})$ such that the logical $Z$ variance is $\text{Var}(\frac{\xi_1^Z+\xi_2^Z}{2})$. In the general case, where the $X$ and $Z$ errors appear simultaneously such as $X(\xi_1^X) Z(\xi_1^Z)\otimes X(\xi_2^X) Z(\xi_2^Z)$, aside from the above analysis, we can use another ancilla to measure the $X$ error acting on the second rotor by measuring stabilizer $\mathbb{I}\otimes Z(\frac{2\pi}{m})$.
The error correction procedure of this code for the general case is to prepare two ancilla rotors in $\ket{+}_{U(1)}$ (one for the $Z$-error and another for $X$-error). For the $X$-error, we can extract the syndrome $\xi_2^X \Mod m$ by applying $\textsc{cnot}_{2\rightarrow 4}$ and measure the angular momentum on the fourth rotor.
The correctable error $X$-error would be in the interval $\xi_2^X \in [-\frac{m}{2}, \frac{m}{2})$. Sharing the spirit of oscillator GKP-stabilizer codes \cite{noh2020encoding}, if we assume the $X$ and $Z$-types errors are identical and independent (i.i.d) random variables and their variances are much smaller than $\frac{\pi}{m}$ and $m$ respectively, then such construction can reduce the variance of $Z$-type error acting on the logical rotor to $\frac{1}{2}$ while it does not amplify the $X$-type error. }  

\

We can repeat the projection procedure we discussed in Sec.~\ref{sec:homo_np_rotor_code} to obtain number-phase rotor resource states for this version of GKP-stabilizer codes.
For this number-phase version, all ingredients are replaced by their number-phase rotor counterparts.
The stabilizer becomes $X(m)^\emph{\dagger}_{\text{np}} \otimes X(m)^\emph{\dagger}_{\text{np}}$ and $\mathbb{I} \otimes Z(\frac{2\pi}{m})$.
The $Z$-type errors correspond to the dephasing channels in bosonic systems, and $X$-type errors (momentum kicks) correspond to the photon loss channel in bosonic systems.

\end{example}

\section{Discussion and Outlook}

In this work, we quantify computational primitives --- states and operations --- of the planar or \(\text{U}(1)\) rotor.
We investigate rotor Clifford operations in Sec.~\ref{sec:rotor_Clifford}. Drawing an analogy to oscillator Gaussian states, we study the rotor nullifier states (analogues of position and momentum eigenstates), coherent states, Josephson junction Hamiltonians, as well as their transformations under the rotor Clifford group in Sec.~\ref{sec:rotor_Gaussian_state}.

Due to the structure of the rotor Clifford group, rotor nullifier states and coherent states are not related via rotor Clifford transformations. These problems call for a suitable definition of the Gaussian and non-Gaussian states for rotors, as well as a quantification of non-Gaussianity, or "magic" \cite{bravyi2005universal, genoni2007measure, howard2017application,chabaud2020stellar,liu2022many}, in rotor states. It is worthwhile to continue this direction since studies of magic and non-Gaussianity of quantum states on discrete and continuous-variable quantum systems yield various results on realizing fault-tolerant quantum computing \cite{kenfack2004negativity,garcia2021bloch,mensen2021phase,hahn2022quantify,baragiola2019all,yamasaki2022cost,liu2022many}, estimation of quantum information resources \cite{veitch2012negative,veitch2014resource,chitambar2019quantum, howard2017application}, and characterizing exotic quantum phases of matter \cite{ellison2021symmetry}.

In Sec.~\ref{sec:preliminary}, we discuss how to faithfully realize a planar rotor, which we call the embedded rotor, inside a harmonic oscillator.
This provides a way to simulate quantum systems described by $\text{U}(1)$ degrees of freedom with harmonic oscillators with modular variables \cite{aharonov1969modular}.
This may offer hardware-efficient quantum simulation schemes of certain $\text{U}(1)$ lattice models, and our newly quantified primitives lend themselves to many-body variational methods.

In Sec.~\ref{sec:preliminary}, we utilize a non-faithful embedding of a rotor inside an oscillator --- the number-phase rotor --- that treats the Fock space of oscillator as the non-negative angular momentum subspace of the rotor.
This enables us to construct oscillator codes against photon number changing and dephasing errors by adapting homological rotor codes in Sec.~\ref{sec:homo_np_rotor_code}.
This yields a multi-mode generalization of number-phase codes \cite{grimsmo2021quantum} which we call homological number-phase codes.
The performance of homological number-phase codes and the number-phase uncertainty relation of these codes are left for future studies. Moreover, the number-phase rotor provides another approach to realize rotor algebras in harmonic oscillator systems.
Its use in quantum simulation of many-body physics is also an interesting question.

In Sec.~\ref{sec:homo_rotor_revisit}, we explain the method of calculating the logical space of homological rotor codes via Smith normal form. Homological rotor codes are classified by their torsion parts, and the torsion structure is invariant under CSS rotor Clifford transformations.

We find a relation between the single-rotor code with torsion and the oscillator GKP qudit in the embedded rotor formalism. Also, we show that rotor GKP codes can be treated as a concatenation between homological rotor codes and modular-qudit GKP codes.

Having quantified some of the basic computational primitives for the planar rotor, we leave open the question of the Clifford hierarchy \cite{gottesman1999quantum,zeng2008semi,cui2017diagonal}
We believe that, once this hierarchy is determined for oscillator systems, projecting each member of the hierarchy into the oscillator's embedded rotor should be useful in backing out the corresponding operators for the planar rotor.

It would be theoretically important as well as practically useful to develop analogous symplectic primitives for more general quantum systems, such as molecules, rigid bodies and other group-valued qudits \eczoo{group_gkp}.
They may open the possibilities of utilizing other physical platforms, as well as providing better candidates for the study of holographic gravity. This would be an interesting direction to pursue in the future.

\begin{acknowledgments}
We thank Christophe Vuillot and Barbara M.\@ Terhal for helpful discussions during the 2023 Quantum Information Processing conference.
Y.X.\@ thanks Yu-An Chen, Ryohei Kobayashi, Zi-Wen Liu, Isaac Kim, and Hassan Shapourian for useful discussions.
Y.W.\@ thanks Sirui Shuai for discussions on semi-direct product group structures.

Y.X.\@ and Y.W.\@ thank the organizers and participants of the YIPQS long-term workshop YITP-T-23-01
"Quantum Information, Quantum Matter and Quantum Gravity (2023)" held at Yukawa Institute for Theoretical Physics in Kyoto University, where part of this work was completed.
Y.X.\@ was partially supported by ARO Grant No. W911NF-15-1-0397,
National Science Foundation QLCI Grant No. OMA2120757, AFOSR-MURI Grant No. FA9550-19-1-0399,
and Department of Energy QSA program.
Y.W.\@ is supported by Shuimu Tsinghua Scholar Program of Tsinghua University.

Contributions to this work by NIST, an agency of the US government, are not subject to US copyright.
Any mention of commercial products does not indicate endorsement by NIST.
Y.X.\@ thanks Yujie Zhang and Shan-Ming Ruan for their hospitality in Tokyo and Kyoto.
Y.X.\@ thanks Yilun Li, Honggeng Zhang, Xuan Su, Qiongsi Yan, and Bo He for their mental support.
V.V.A.\@ thanks Olga Albert and Ryhor Kandratsenia for providing daycare support throughout this work.
\end{acknowledgments}

\appendix

\section{Wigner Function of rotor GKP codes}
\label{app:coherent_Wigner}
In the oscillator cases, GKP states are non-Gaussian states that have Wigner negativity and infinite Stellar rank \cite{walschaers2021non,mensen2021phase,garcia2021bloch,hahn2022quantify,chabaud2020stellar,chabaud2022holomorphic}. The studies of the phase space structure and non-Gaussianity of bosonic modes yield various discoveries of classical simulability and quantum magic of bosonic systems \cite{baragiola2019all,chabaud2020stellar, garcia2020efficient,walschaers2020practical,chabaud2021continuous,bourassa2021fast,chabaud2022holomorphic,hahn2022quantify,calcluth2022efficient,chabaud2023resources}.

In this section, we calculate the Wigner function of rotor GKP code as an example, and show its Wigner function has a similar form as the Wigner function of oscillator GKP states, and indeed, the rotor GKP states have Wigner negativity and can be written as a sum of Kronecker and Dirac delta functions. The presence of Kronecker delta functions is because of the discrete-variable nature of angular momentum, while, in contrast, both position and momentum are continuous variables in harmonic oscillators.

The studies of the Wigner function of the planar rotor are initialized in the context of orbital angular momentum states of light~\cite{rigas2008full,rigas2010non,rigas2011orbital,kastrup2016wigner,kastrup2017wigner_operator,kastrup2017wigner_qi}, which is also described by planar rotor. The modular Wigner function is also studied in the context of oscillator GKP codes, considering the modular nature of position and momentum variables \cite{fabre2020wigner}.

The definition of Wigner function of planar rotor is written as
\begin{eqs}
    W_\rho(l,\phi)=\frac{1}{2\pi} \int_{-\pi}^{\pi} \bra{\phi-\frac{\phi'}{2}} \rho \ket{\phi+\frac{\phi'}{2}} e^{i\phi'l}d\phi',
\end{eqs}
where the Wigner function has two canonical variable: phase $\phi$ and angular momentum $l$. This function is defined on a infinite cylinder $\mathbb{T}\times \mathbb{Z}$ as shown in Fig.~\ref{fig:phase_space}.

Here we consider $\rho=\ket{0}_L \bra{0}_L$ and calculate its Wigner function
\begin{eqs}\label{eq:rotor_wigner}
    W_\rho(l,\phi)=&\frac{1}{2\pi} \int_{-\pi}^{\pi} \sum_{m,m'=0}^{N-1} \langle \phi-\frac{\phi'}{2}|\frac{2\pi m}{N}\rangle \langle \frac{2\pi m'}{N}| \phi+\frac{\phi'}{2}\rangle e^{i\phi'l}d\phi'\\
    =&\frac{1}{2\pi} \int_{-\pi}^{\pi} \sum_{m,m'=0}^{N-1} \delta_{2\pi} (\phi-\frac{\phi'}{2}-\frac{2\pi m}{N}) \\
    &\delta_{2\pi} (\phi+\frac{\phi'}{2}-\frac{2\pi m}{N}) e^{i\phi' l}d\phi'.
\end{eqs}
Here the $\delta_{2\pi}(x)$ represents periodic delta function, which is
\begin{eqnarray}
		\delta_{2\pi}(x) = \left\{
		\begin{aligned}
			&\delta(0) \quad &\text{if}\quad  x\Mod 2\pi=0\\
			&0 &\quad \text{otherwise}.
		\end{aligned}
		\right.
	\end{eqnarray}

Hence, we can rewrite Eq.~\eqref{eq:rotor_wigner} as
\begin{eqs}
    &W_\rho(l,\phi)\\
    =&\frac{1}{2\pi} \int_{-\pi}^{\pi} \sum_{m,m'\in \mathbb{Z}} \delta (\phi-\frac{\phi'}{2}-\frac{2\pi m}{N}) \delta (\phi+\frac{\phi'}{2}-\frac{2\pi m'}{N}) e^{i\phi' l}d\phi'\\
    \propto & \frac{1}{2\pi} \sum_{m,m'\in \mathbb{Z}} \delta(2\phi -\frac{2\pi}{N}(m+m')) e^{il(-2\phi+\frac{4\pi m'}{N})}\\
    = & \frac{1}{ 2\pi} \sum_{c,d\in \mathbb{Z}} \delta(\phi-\frac{\pi c}{N}) (-1)^{ cd} \delta_{l,Nd/2}.
\end{eqs}

Similarly, the Wigner function of $\ket{1} \bra{1}$ is
\begin{eqs}
    W_{\ket{1}\bra{1}} (l,\phi) \propto \frac{1}{ 2\pi} \sum_{c,d\in \mathbb{Z}} \delta(\phi-\frac{\pi (c+1)}{N}) (-1)^{ cd} \delta_{l,Nd/2}.
\end{eqs}
The Wigner function of rotor GKP states shows strong negativity relative to its oscillator counterparts. The difference between the Wigner functions of oscillator and rotor GKP states is that the former is a sum of products of Dirac delta functions in both position and momentum, while the latter is a sum of products of Dirac delta functions in phase and Kronecker deltas in angular momentum.

Another interesting phenomenon is that, in the rotor case, the angular momentum eigenstates $\ket{l}, l\in \mathbb{Z}$ are the only normalizable states with non-negative Wigner function~\cite{rigas2010non}. This statement matches our understanding of oscillator systems with discrete translational symmetry in position direction $S_q=e^{i2\sqrt{\pi}\hat{p}}$. The oscillator system with spatial periodicity will have quantized angular momentum defined on $\mathbb{Z}$, and it can be regarded as a rotor (mathematically). In this picture, the angular momentum eigenstates of the rotor correspond to the momentum eigenstates of a periodic oscillator, which are Gaussian states with non-negative Wigner functions.

\section{Quantum Error-Correction Condition for Normalized Rotor GKP codes}\label{app:ec_condition_rotor_GKP}

For the ideal oscillator GKP states, the codewords are equal-weight superpositions of infinite numbers of Delta functions which are unnormalizable and require unbounded energy to prepare.
In practice, we typically impose various regularizers to impose the normalization and finite-energy conditions, and the Gaussian regularizer can be implemented experimentally in trapped-ion, superconducting-circuit, and optical platforms \cite{fluhmann2019encoding,royer2020stabilization,de2022error, sivak2022model,sivak2023real,konno2023propagating}.
In this section, we study the quantum error correction conditions \cite{knill1997theory,bennett1996mixed} for normalized rotor GKP states that are regularized by the Gaussian regularizer $E_\Delta (\hat{l})$.

The unnormalizable codeword of rotor GKP code is
\begin{eqs}
    \ket{0}_L=\sum_{k\in \mathbb{Z}} \ket{l=kN}, \quad \ket{1}_L=\sum_{k\in \mathbb{Z}} (-1)^k \ket{l=kN}.
\end{eqs}
To normalize them, we impose a Gaussian envelope $e^{-\Delta \hat{l}^2}$ such that
\begin{eqs}
    &\ket{0}_\Delta= \sum_{k\in \mathbb{Z}} e^{-\Delta (kN)^2} \ket{l=kN}, \\
    &\ket{1}_\Delta=\sum_{k \in \mathbb{Z}} e^{-\Delta (kN)^2} (-1)^k \ket{l=kN}.
\end{eqs}

The error operator is written as $E_m(\theta)=Z(\theta) X(m)$. Hence, we can calculate the quantum error-correction condition for the normalized rotor GKP states
\begin{eqs}\label{eq:QECC_normalized_rotor_GKP}
    \bra{\phi_i} E_{m'}(\theta')^\dagger E_m (\theta) \ket{\phi_j}= e^{i (\theta-\theta')m}\bra{\phi_i}X(m-m') Z(\theta-\theta') \ket{\phi_j}.
\end{eqs}
\begin{widetext}
We have
\begin{eqs}
    X(m-m')Z(\theta-\theta') \ket{0}_\Delta &= \sum_{k\in \mathbb{Z}} e^{-\Delta (kN)^2} e^{i(\theta-\theta')kN} \ket{l=kN+m-m'},\\
    X(m-m')Z(\theta-\theta') \ket{1}_\Delta &= \sum_{k\in \mathbb{Z}} (-1)^k e^{-\Delta (kN)^2} e^{i(\theta-\theta')kN} \ket{l=kN+m-m'}.
 \end{eqs}
 Then we calculate Eq.~\eqref{eq:QECC_normalized_rotor_GKP} utilizing Jacobi theta functions
 \begin{eqs}
     \text{(a)} \quad & e^{i (\theta-\theta')m}\bra{0}_\Delta X(m-m') Z(\theta-\theta') \ket{0}_\Delta= e^{i(\theta-\theta')m}\sum_{k,k'\in \mathbb{Z}} \delta_{k'N, kN+m-m'} e^{-\Delta N^2 (k^2+ k'^2)} e^{i(\theta-\theta')kN}\\
     =& e^{i(\theta-\theta')m} \sum_{k\in \mathbb{Z}} e^{-\Delta N^2 (k^2 + (k+\frac{m-m'}{N})^2)} e^{i (\theta-\theta')kN} \delta_{m-m' ~\text{mod}~ N,0}\\
     =& e^{i(\theta-\theta')\frac{m+m'}{2}} e^{-\Delta \frac{(m-m')^2}{2}} \sum_{k\in \mathbb{Z}} e^{-2\Delta N^2 (k+\frac{m'-m}{2N})^2} e^{i(\theta-\theta')N (k+\frac{m'-m}{2N})} \delta_{m-m' ~\text{mod}~ N,0}\\
     =& \begin{cases}  e^{i(\theta-\theta')\frac{m+m'}{2}} e^{-\Delta \frac{(m-m')^2}{2}} \delta_{m-m' ~\text{mod}~ N,0} \vartheta_2 (z=\frac{(\theta-\theta')N}{2},q= e^{-2 \Delta N^2}) \quad \text{if}~~ \frac{m'-m}{N} ~~\text{is odd},  \\
     e^{i(\theta-\theta')\frac{m+m'}{2}} e^{-\Delta \frac{(m-m')^2}{2}} \delta_{m-m' ~\text{mod}~ N,0} \vartheta_3 (z=\frac{(\theta-\theta')N}{2},q= e^{-2 \Delta N^2}) \quad \text{if}~~ \frac{m'-m}{N} ~~\text{is even}.
     \end{cases}\\
     \text{(b)} \quad & e^{i (\theta-\theta')m}\bra{1}_\Delta X(m-m') Z(\theta-\theta') \ket{1}_\Delta \\
     =&\begin{cases} -e^{i(\theta-\theta')\frac{m+m'}{2}} e^{-\Delta \frac{(m-m')^2}{2}} \delta_{m-m' ~\text{mod}~ N,0} \vartheta_2 (z=\frac{(\theta-\theta')N}{2},q= e^{-2 \Delta N^2}) \quad \text{if}~~ \frac{m'-m}{N} ~~\text{is odd},\\
     e^{i(\theta-\theta')\frac{m+m'}{2}} e^{-\Delta \frac{(m-m')^2}{2}} \delta_{m-m' ~\text{mod}~ N,0} \vartheta_3 (z=\frac{(\theta-\theta')N}{2},q= e^{-2 \Delta N^2}) \quad \text{if}~~ \frac{m'-m}{N} ~~\text{is even}
     \end{cases}\\
     \text{(c)} \quad & e^{i (\theta-\theta')m}\bra{1}_\Delta X(m-m') Z(\theta-\theta') \ket{0}_\Delta \\
     = & \begin{cases} -e^{i(\theta-\theta')\frac{m+m'}{2}} e^{-\Delta \frac{(m-m')^2}{2}} \delta_{m-m' ~\text{mod}~ N,0} \vartheta_1 (z=\frac{(\theta-\theta')N}{2},q= e^{-2 \Delta N^2}) \quad \text{if}~~ \frac{m'-m}{N} ~~\text{is odd},\\
     e^{i(\theta-\theta')\frac{m+m'}{2}} e^{-\Delta \frac{(m-m')^2}{2}} \delta_{m-m' ~\text{mod}~ N,0} \vartheta_4 (z=\frac{(\theta-\theta')N}{2},q= e^{-2 \Delta N^2}) \quad \text{if}~~ \frac{m'-m}{N} ~~\text{is even}
     \end{cases}~.
 \end{eqs}
 \end{widetext}

 \section{Discussions on the no-go theorem for oscillator}\label{app:nogo_rotor}

In Ref.~\cite{vuillot2019toric}, authors proved a no-go theorem for Gaussian-stabilizer codes. The statement is: for mode-to-mode codes, if the encoding, error correction, and decoding all consist of only Gaussian operations, then these codes cannot correct Gaussian quadrature displacement errors.
Suppose the logical quadrature errors follow a Gaussian distribution $\mathcal{N}(0,\sigma^2_{q/p})$, then $\sigma_{q_L}^2 \sigma_{p_L}^2=\sigma_q^2 \sigma_p^2$ after encoding/decoding. This no-go theorem indicates that Gaussian stabilizer codes can only rotate/squeeze Gaussian errors, but will never reduce variance on both quadratures.

In this section, we will briefly review the derivation of Gaussian no-go theorem and its limitation, then we will comment its relevance to homological rotor codes.

We first state the conditions for the no-go theorem to be true:
\begin{itemize}
    \item Encoding unitary is a Gaussian operation (symplectic transformation), and ancilla states are all initialized in infinitely squeezed states (Gaussian states).

    \item The error correction is adding linear combinations of nullifiers onto logical quadratures. For example, the maximum-likelihood error correction is adding $-CG^T (GG^T)^{-1} G$ onto the logical quadrature $C$.

    \item Quadratures are defined on $\mathbb{R}$.
\end{itemize}

The derivation utilizes the linearity and orthogonality of symplectic vectors. Although the analog rotor codes also share a symplectic structure, their phase quadrature is a modular quadrature which doesn't have linearity. The lack of linearity in rotor systems provides an obstruction to generalizing the Gaussian no-go theorem for oscillators to rotor systems.

For a $[[n,k,d]]$ Gaussian stabilizer codes, the Gaussian unitary encoders $U_{\text{enc}}$ are symplectic transformations
\begin{eqs}
    U_{\text{enc}} \Vec{r} U_{\text{enc}}^\dagger=U_{\text{enc}} (\hat{q}_1,...,\hat{q}_n,\hat{p}_1,...,\hat{p}_n)^T U_{\text{enc}}^\dagger = A \Vec{r}.
\end{eqs}
We can decompose the symplectic matrix $A$ as \cite{vuillot2019toric}
\begin{eqs}
    A=\begin{pmatrix}
        Q\\
        G\\
        P\\
        D
    \end{pmatrix}.
\end{eqs}

The syndrome $\Vec{z}$ is given by following equation
\begin{eqs} \label{eq:nullifier_syndrome}
    \Vec{z}=G\Vec{\xi},
\end{eqs}
where $\Vec{\xi}$ is a $2n$-dimensional noise vector. The error-corrected logical quadrature can be written as
\begin{eqs} \label{eq:nullifier_QEC}
    C'=C-CG^T (GG^T)^{-1} G=C+\Lambda G.
\end{eqs}
And the covariance matrix of error-corrected logical quadratures can be diagonalized
\begin{eqs}
   K(C'C'^T)K^{-1}=\begin{pmatrix}
        \text{diag}(\sigma_{q,j}^2) & 0\\
       0 & \text{diag}(\sigma_{p,j}^2).
   \end{pmatrix}
\end{eqs}

This no-go theorem is applicable once Eqs.~\eqref{eq:nullifier_syndrome} and \eqref{eq:nullifier_QEC} are linear. However, in rotor case, Eq.~\eqref{eq:nullifier_syndrome} is no longer linear, the rotor syndrome has modular structure,
\begin{eqs}
    \Vec{z}_{\text{rotor}}= R_{2\pi } (G \Vec{\xi}),
\end{eqs}
where $R_{2\pi}$ is a rounding function that rounds the input to the nearest mulitplicity of $2\pi$. The modular structure is non-linear, hence, the no-go theorem will not hold true in rotor systems. However, if we drop the modularity by assuming the variance of syndrome is much smaller than $2\pi$, the rotors will be reduced to regular oscillators where the Gaussian no-go theorems holds true.

\bibliographystyle{ieeetr}
\bibliography{biblo.bib}

\end{document}